\documentclass[12pt]{article}
\usepackage{amsmath}
\usepackage{graphicx}
\usepackage{enumerate}
\usepackage{natbib}
\usepackage{url} % not crucial - just used below for the URL 
\usepackage{xcolor}
\usepackage{enumitem}
\usepackage{subcaption}
\usepackage{makecell}
\newcommand{\blind}{1}

\usepackage{amsthm}
\newtheorem{theorem}{Theorem}[section]
\newtheorem{lemma}[theorem]{Lemma}
\newtheorem{corollary}[theorem]{Corollary}

\newtheorem{proposition}[theorem]{Proposition}
\newtheorem{remark}[theorem]{Remark}
\newtheorem{assumption}[theorem]{Assumption}
\newtheorem{example}[theorem]{Example}

\usepackage{float}
\usepackage{mathrsfs}
\usepackage{amsfonts}
\usepackage{bbm}
\usepackage{bm}
\usepackage{amsmath}
\usepackage{amssymb}
\usepackage{graphicx}
\usepackage{tabularx}
\usepackage{float}
\usepackage{appendix}
\usepackage{multicol}
\usepackage{multirow}

\begin{document}

\def\spacingset#1{\renewcommand{\baselinestretch}%
{#1}\small\normalsize} \spacingset{1}

%%%%%%%%%%%%%%%%%%%%%%%%%%%%%%%%%%%%%%%%%%%%%%%%%%%%%%%%%%%%%%%%%%%%%%%%%%%%%%

\if1\blind
{
  \title{\bf Adversarially  Perturbed Precision Matrix Estimation}
  \author{Yiling Xie\\
    Department of Decision Analytics and Operations,\\  City University of Hong Kong.}
  \maketitle
} \fi

\if0\blind
{
  \bigskip
  \bigskip
  \bigskip
  \begin{center}
    {\LARGE\bf Adversarially Robust Precision Matrix Estimation}
\end{center}
  \medskip
} \fi

\bigskip

\sloppy 
\begin{abstract}%   <- trailing '%' for backward 
Precision matrix estimation is a fundamental topic in multivariate statistics and
modern machine learning. This paper proposes an adversarially perturbed precision
matrix estimation framework, motivated by recent developments in adversarial training.
The proposed framework is versatile for the precision matrix problem since, by adapting
to different perturbation geometries, the proposed framework can not only recover
the existing distributionally robust method but also achieve high-dimensional model selection consistency under the scale-adaptive incoherence condition, which can be viewed as a relaxation of the classic incoherence condition in the heteroscedastic settings.
Additionally, the proposed perturbed precision matrix estimation
framework is asymptotically equivalent to the regularized precision matrix estimation,
and the asymptotic normality can be established accordingly, where the asymptotic bias introduced by perturbation is highlighted. 
Numerical experiments demonstrate the
desirable practical performance of the proposed adversarially perturbed approach.
\end{abstract}

\spacingset{1.9}
\section{Introduction}\label{sectionintro}
The precision matrix, defined as the inverse of the covariance matrix, characterizes the conditional dependence structure of the data, where a zero off-diagonal entry implies that two variables are conditionally independent given all other variables \citep{fan2016overview}.
Estimating the precision matrix is essential in many problems, including neuroimaging \citep{huang2010learning}, network analysis \citep{torri2018robust}, portfolio optimization \citep{procacci2022portfolio}, and genomics \citep{schafer2005shrinkage}, among others.

Suppose $\bm{x}^1, \ldots, \bm{x}^n$ are $n$ independent and identically distributed samples drawn from a $d$-dimensional multivariate normal distribution. For simplicity, we assume that the population mean is zero and the population covariance matrix $\Sigma$ is positive definite and unknown.
The standard approach to obtain the precision matrix estimation, i.e., $\Sigma^{-1}$, is to identify the minimizer of 
the negative log-likelihood function:  
%(\textcolor{red}{observations are centereed})
\begin{equation}\label{orig}\min_{ C\succ 0} -\log \mathrm{det}  C + \mathbb{E}_{\mathbb{P}_n} [\bm{x}^\top  C \bm{x}],\end{equation}
where $ C\succ 0$ means that $ C$ is positive definite, $\mathrm{det} C$ denotes the determinant of the matrix $ C$, $\mathbb{P}_n$ is the empirical distribution.
When the matrix dimension is small such that $d<n$, the solution to problem \eqref{orig} is the inverse of the sample covariance matrix $\mathbb{E}_{\mathbb{P}_n} [\bm{x}\bm{x}^\top]$. However, if the matrix dimension is large such that $d>n$, the problem \eqref{orig} does not admit a solution.
Numerous studies have been conducted to explore shrinkage and regularization frameworks for precision matrix to achieve estimations which are invertible, well-conditioned, and enjoy desirable statistical properties such as sparsity \citep{yuan2007model,ledoit2012nonlinear,nguyen2022distributionally}. 
We will show later that the proposed adversarially perturbed framework in this paper can recover some of these shrinkage and regularization effects by selecting different types of perturbations.

Our proposed adversarially perturbed precision matrix framework is inspired by recent developments in trustworthy machine-learning---adversarial training \citep{Goodfellow2015}. The perturbed estimation is obtained by solving the following optimization problem:
\begin{equation}\label{intro}
\inf_{ C\succ 0}\left\{ -\log \mathrm{det}  C +\mathbb{E}_{\mathbb{P}_n} \left[\max_{\Vert \bm{\Delta}\Vert\leq \delta}(\bm{x}+\bm{\Delta})^\top  C (\bm{x}+\bm{\Delta})\right]\right\},\end{equation}
where $\Vert\cdot\Vert$ is the norm of the perturbation $\bm{\Delta}$, and $\delta$ is the magnitude of the perturbation. 
Instead of optimizing the empirical average of  the likelihood function, the adversarially perturbed estimation problem \eqref{intro} seeks for the minimizer of the loss under the most adversarial perturbation $\bm{\Delta}$, which is determined by both $C$ and $\bm{x}$ through the inner maximization problem in \eqref{intro}, subject to the constraint that the norm of the perturbation is bounded.
%The perturbation $\bm{\Delta}$ is imposed to each realization of $\bm{x}$.

In this paper, we will focus on the $\ell_p$-norm of the perturbation. 

When $p=2$, it can be proven that the perturbed estimation problem \eqref{intro} is equivalent to the Wasserstein shrinkage estimation proposed in \cite{nguyen2022distributionally}. The resulting optimization problem is tractable and yields a nonlinear shrinkage estimator that is well-conditioned, rotation equivalent, and preserves the eigenvalue order \citep{nguyen2022distributionally}. 

When $p=\infty$, we optimize a convex upper-bound surrogate of the objective function in problem \eqref{intro} to admit tractability. 
The resulting estimation is scale-adaptive and proven to achieve adversarial robustness and high-dimensional model selection consistency.
The high-dimensional model selection consistency achieved by the $\ell_1$-regularized estimation requires the mutual incoherence condition \citep{raskutti2008model,wainwright2009sharp,ravikumar2011high}. 
However, our $\ell_\infty$-perturbed precision matrix estimation reduces the upper bound of the mutual incoherence by a factor of the ratio of the maximum to the minimum marginal standard deviation of the variables. This property facilitates consistency across a broader range of scenarios, particularly under heteroscedastic settings.

Then, we demonstrate that introducing perturbations is asymptotically equivalent to imposing regularization in precision matrix estimation, provided the perturbation magnitude is small. 
% The corresponding first-order and second-order regularization terms are also given, seeing Section \ref{regeffect}.
Based on the regularization effect, the asymptotic distribution of the adversarially perturbed estimation is derived. The asymptotic distribution characterizes the bias induced by the inner maximization problem in our proposed framework \eqref{intro}. We suppose the perturbation magnitude follows the rule $\delta=\eta n^{-\gamma}$, where $\eta > 0$ and $\gamma > 0$. 
Our analysis yields the following two regimes: When $\gamma > 1/2$, the asymptotic bias is zero. In this regime, the adversarial perturbation decays sufficiently fast such that it does not affect estimation accuracy. When $\gamma \leq 1/2$, the asymptotic distribution has a non-zero bias. 
Furthermore, specifically for $\ell_\infty$-perturbed surrogate problem, the derived asymptotic distribution demonstrates that sparsity recovery is achievable when $\gamma = 1/2$. 
% The regularization reformulation and asymptotic distribution provide insights into the statistical behavior of both Wasserstein shrinkage estimation and the sparse estimation recovered by our novel framework.
See a summary table of the associated asymptotic behavior in Table \ref{tab}, Section \ref{asynsec}.
\subsection{Related Work}
\textit{\textbf{Adversarial Training}}:

 As the modern artificial intelligence and machine learning models grow more powerful and widely used, their reliability becomes increasingly important. One major challenge to reliability is that machine learning models can be vulnerable to small changes in input data, which may lead to unreliable outputs \citep{Goodfellow2015}.
To improve the model robustness, adversarial training has been introduced as a principled approach by optimizing against perturbations posed in the input data \citep{Goodfellow2015, madry2018towards}.

Different from the classic empirical risk minimization approach, the resulting adversarial training problem is a min-max robust optimization problem  and thus introduces a new statistical estimation paradigm. 
While most of the relevant existing work focus on the vector-valued parameter in Euclidean space \citep{Goodfellow2015, madry2018towards,javanmard2022precise,ribeiro2023regularization,xie2024high,xie2025asymptotic}, our proposed framework illustrates how adversarial training inspires a novel methodology for estimating matrix-valued parameters within the positive-definite cone.
To the best of our knowledge, we are the first to formulate the precision matrix estimation problem based on adversarial training.
% The proposed framework serves as an alternative for practitioners when the precision matrix estimation with desirable properties such as sparsity and adversarial robustness is needed.

Adversarial training is closely related to distributionally robust optimization. It is well known that adversarial training can be viewed as a special case of the Wasserstein distributionally robust optimization \citep{staib2017distributionally,gao2024wasserstein}. This paper, however, provides a reverse perspective by demonstrating that, in the context of precision matrix estimation, the Wasserstein distributionally robust optimization is a special case of adversarial training.

\textit{\textbf{Precision Matrix Estimation}}:

Precision matrix estimation has been extensively studied under different settings. Classical approaches primarily focus on enforcing sparsity via $\ell_1$-regularization or $\ell_1$-minimization \citep{yuan2007model,friedman2008sparse,cai2011constrained}.
The novel framework proposed in this paper is also capable of the sparsity recovery and model selection consistency. 
However, as mentioned in Section \ref{sectionintro} and detailed in Section \ref{varianceadaptive}, the model consistency selection achieved by our framework is based on the scale-adaptive incoherence condition which can be weaker than that for the $\ell_1$-regularized estimation, especially under the heteroscedastic settings. The reason is that the adaptive incoherence in our framework has a smaller upper bound, potentially relaxing the requirements for model
consistency selection when variances differ significantly.

As applications have broadened, researchers have developed methods robust to data anomalies. For instance, \cite{loh2018high} establishes a framework robust to data contamination, \cite{tang2021robust} focuses on mitigating the impact of outliers, \cite{avella2018robust} studies robustness against heavy-tailed distributions, and the method framed in \cite{nguyen2022distributionally} addresses distributional shifts.
From the perspective of robustness, 
our framework can recover the Wasserstein robust estimation \citep{nguyen2022distributionally} under $\ell_2$-perturbation.
In addition, by defining the model through an adversarial training lens, we achieve intrinsic adversarial robustness. Our aforementioned scale adaptivity can also be interpreted as the robustness towards the  heteroscedasticity.
\subsection{Notations and Definitions}
$X_{ij}$ denotes the $(i,j)$-th entry of the matrix $X$.
$X_{\mathcal{I}\mathcal{J}}$ denotes the submatrix of $X$ defined by the row index set $\mathcal{I}$ and the column index set $\mathcal{J}$. $\mathrm{sign}(x)$ denotes the sign of a scalar $x$ with $\mathrm{sign}(0)=0$. Accordingly, $\mathrm{sign}(X)$ denotes the element-wise application of the sign function to matrix $X$.
$\Vert\bm{x}\Vert_p$ denotes the $\ell_p$-norm of the vector $\bm{x} \in \mathbb{R}^d$, defined as $\Vert\bm{x}\Vert_p = (\sum_{j=1}^d\vert x_j \vert^p)^{1/p}$ for $1\leq p<\infty$, and $\Vert\bm{x}\Vert_\infty=\max_{1\leq j\leq d} \vert x_j\vert$.
$\Vert X \Vert_\infty$ denotes the induced matrix $\ell_\infty$-norm, defined as $\Vert X \Vert_\infty = \max_{1\leq i\leq d} \sum_{j=1}^d \vert X_{ij} \vert$.
$\Vert X \Vert_{1,1}$ denotes the element-wise $\ell_1$-norm, defined as $\Vert X \Vert_{1,1} = \sum_{i=1}^d\sum_{j=1}^d \vert [X]_{ij} \vert$.
$\Vert X \Vert_{\max}$ denotes the maximum element-wise absolute value, defined as $\Vert X \Vert_{\max} = \max_{_{1\leq i,j\leq d}}\vert X_{ij}\vert$.
$\odot$ denotes the Hadamard product.
$\otimes$ denotes the Kronecker product.
$\mathbb{I}(\cdot)$ denotes the indicator function.
$\mathrm{vec}(X)$ denotes the vectorization of the matrix $X$.
\subsection{Organization of this Paper}
The remainder of this paper is organized as follows. 
In Section \ref{ell2}, we study the adversarially $\ell_2$-perturbed precision matrix estimation and its equivalence to the Wasserstein shrinkage estimation.
In Section \ref{ellinf}, we investigate the formulation and the properties of the adversarially $\ell_\infty$-perturbed precision matrix estimation.
In Section \ref{regeffect}, we introduce the regularization effect of the proposed framework.
In Section \ref{asynsec}, we study the asymptotic distribution of the proposed precision matrix estimation.
Numerical experiments are conducted and analyzed in Section \ref{numsection}.
Future work is discussed in Section \ref{dis}.
The proofs are relegated to the appendix whenever possible.

\section{Adversarially $\ell_2$-perturbed Precision Matrix Estimation}\label{ell2}
In this section, we will discuss the $\ell_2$-perturbed precision matrix estimation.

The adversarially perturbed precision matrix estimation under $\ell_2$-norm is the solution to the following optimization problem:
\begin{equation}\label{mainproblem1}
\inf_{ C\succ 0}\left\{ -\log \mathrm{det}  C +\mathbb{E}_{\mathbb{P}_n}\left[\max_{\Vert \bm{\Delta}\Vert_2\leq \delta}(\bm{x}+\bm{\Delta})^\top  C (\bm{x}+\bm{\Delta})\right]\right\},\end{equation}
where the perturbation norm in the constraint set in the inner maximization problem is chosen as $\ell_2$-norm.
\subsection{Tractable Reformulation}
In this subsection, we give the convex tractable reformulations of the adversarially $\ell_2$-perturbed precision matrix estimation.
\begin{theorem}[Convex Reformulations under $\ell_2$-perturbation] \label{tractable}
    The problem \eqref{mainproblem1} is equivalent to the following problem:
    \begin{equation}\label{equiproblem}
\inf_{C\succ 0,\lambda I - C \succ 0}\left\{ -\log \mathrm{det}  C+\mathbb{E}_{\mathbb{P}_n}[\bm{x}^\top  C\bm{x}] +  \lambda\delta^2+\mathbb{E}_{\mathbb{P}_n}[\bm{x}^\top  C (\lambda I - C )^{-1}  C\bm {x}] \right\},\end{equation}
which further admits the following equivalent reformulation: 
\begin{equation}\label{detsemidefinite}
\begin{aligned}
\inf_{\lambda, \{t_i\}, C} \quad &  -\log \mathrm{det}  C+\frac{1}{n}\sum_{i=1}^n{\bm{x}^i}^\top  C\bm{x}^i+\lambda\delta^2+\frac{1}{n}\sum_{i=1}^n t_i \\
\mathrm{subject\ to} \quad & C\succ 0,\lambda I - C \succ 0,\\
& \begin{pmatrix}
t_i &  {\bm{x}^i}^\top C \\
C\bm {x}_i & \lambda I - C 
\end{pmatrix} \succeq 0, \quad \forall i = 1, \dots, n.\\
\end{aligned}\end{equation}
\end{theorem}

The reformulation \eqref{equiproblem} is obtained by deriving the dual of the inner maximization problem in \eqref{mainproblem1}, which admits a closed-form solution under $\ell_2$-perturbation. Subsequently, applying the Schur complement yields formulation \eqref{detsemidefinite}. 

Problem \eqref{detsemidefinite} is a convex semidefinite optimization problem with a log-determinant barrier term in the objective function and linear matrix inequalities in the decision variables $\lambda, \{t_i\}, C$.
Such problems fall within the scope of semidefinite (conic) programming and can be solved efficiently by interior-point methods \citep{boyd2004convex,nemirovski2007advances}.
\subsection{Equivalence to the Wasserstein Shrinkage Estimation}\label{equivalencesubsection}
In this subsection, we will show that the $\ell_2$-perturbed precision matrix estimation problem is equivalent to the Wasserstein distributionally robust precision matrix estimation proposed in \cite{nguyen2022distributionally}, which they coin as the Wasserstein shrinkage estimation.

The optimization problem for obtaining the Wasserstein shrinkage estimation introduced in \cite{nguyen2022distributionally} can be written as the following problem:
\begin{equation}\label{WDRO}
\inf_{ C\succ 0}\left\{ -\log \mathrm{det}  C + \sup_{S\in \mathcal{U}_{\bar{A},\rho}}\mathrm{tr}(S C), \right\},\end{equation}
where $\mathcal{U}_{\bar{A},\rho}$ is the uncertainty set and $\rho$ is the radius of the set. The uncertainty set is centered at the sample covariance $\bar{A}=\mathbb{E}_{\mathbb{P}_n}[\bm{x}\bm{x}^\top]$ and is defined in terms of the metric induced by the 2-Wasserstein distance on the family of normal distributions, and we have that 
\[\mathcal{U}_{\bar{A},\rho} = \left\{ S\succeq 0: \mathrm{tr}(\bar{A})+\mathrm{tr}(S)-2\mathrm{tr}(\sqrt{\bar{A}^{\frac{1}{2}}S\bar{A}^{\frac{1}{2}}})\leq \rho^2\right\}.\]

\cite{nguyen2022distributionally} shows that the min-max optimization problem in \eqref{WDRO} admits a convex reformulation, as restated in Theorem \ref{nguyenreformulation}.
\begin{theorem}[\cite{nguyen2022distributionally}]\label{nguyenreformulation}
    For any fixed $\rho>0$, the  Wasserstein distributionally robust estimation problem \eqref{WDRO} is equivalent to 
        \begin{equation}\label{WDROform}
\inf_{C\succ 0,\lambda I - C \succ 0}\left\{ -\log \mathrm{det}  C+ \lambda \left(\rho^2- \frac{1}{n}\sum_{i=1}^n\bm{x}^i {\bm{x}^i}^\top \right)+\lambda^2 \frac{1}{n}\sum_{i=1}^n\bm{x}^i (\lambda I- C  )^{-1}{\bm{x}^i}^\top   \right\}.\end{equation}
\end{theorem}

The equivalence of our adversarially perturbed estimation problem \eqref{mainproblem1} and distributionally robust trained problem \eqref{WDRO} can be built upon problem \eqref{equiproblem} and \eqref{WDROform} based on algebraic calculations. We state this result in the following theorem.
\begin{theorem}[Equivalence]\label{equivalence}
For fixed $\rho=\delta$, the $\ell_2$-perturbed precision matrix estimation problem \eqref{mainproblem1} is equivalent to the Wasserstein distributionally robust precision matrix estimation problem \eqref{WDRO}.
\end{theorem}

Theorem \ref{equivalence} implies that,  in the precision matrix estimation problem,  our proposed adversarially perturbed approach could recover the Wasserstein distributionally estimator if the $\ell_2$-perturbation is chosen.
Imposing an $\ell_2$-perturbation to the observation $\bm{x}$ is equivalent to constructing an ambiguity set for the sample covariance matrix in terms of the Wasserstein distance.
As illustrated in \cite{nguyen2022distributionally}, the resulting estimator is invertible, rotation equivalent, preserves the eigenvalue order, and can be obtained by solving a sequential quadratic approximation algorithm.

Although arriving at the same solution, our proposed perturbed framework \eqref{mainproblem1} adopts a sample-wise perturbation perspective and admits the asymptotic expansion, from which the asymptotic distribution is derived, seeing Section \ref{regeffect} and Section \ref{asynsec}.
In this way, our framework also helps facilitate rigorous statistical inference.

\section{Adversarially $\ell_\infty$-perturbed Precision Matrix Estimation}\label{ellinf}
%This subsection gives a different formulation.
In this section, we explore the adversarially $\ell_\infty$-perturbed precision matrix estimation.

For the adversarially perturbed precision matrix estimation \eqref{intro}, if the norm of the perturbation is chosen as $\ell_\infty$-norm, we have the following formulation:
\begin{equation}
\label{mainproblem2}
\inf_{ C\succ 0}\left\{ -\log \mathrm{det}  C +\mathbb{E}_{\mathbb{P}_n}\left[\max_{\Vert \bm{\Delta}\Vert_\infty\leq \delta}(\bm{x}+\bm{\Delta})^\top  C (\bm{x}+\bm{\Delta})\right]\right\}.\end{equation}

The inner maximization problem in problem \eqref{mainproblem2} is NP-hard \citep{nemirovski1999maximization}.
To ensure computational tractability, we replace the NP-hard inner maximization in \eqref{mainproblem2} with a convex upper-bound surrogate. Based on the triangle inequality on the perturbation terms, we derive the following formulation:
% \begin{equation*}\begin{aligned}-\log \mathrm{det}  C + \mathbb{E}_{\mathbb{P}_n}\left[\bm{x}^\top  C\bm{x}\right]+ 2 \delta \mathbb{E}_{\mathbb{P}_n}\left[ \Vert  C\bm {x}\Vert_1\right]\end{aligned}\end{equation*}
% If we formulate our problem as follows:
\begin{equation}\label{adaptive}\inf_{ C\succ 0}\left\{-\log \mathrm{det}  C + \mathbb{E}_{\mathbb{P}_n}\left[\bm{x}^\top  C\bm{x}\right]+ 2\delta\sum_{j=1}^d \sum_{k=1}^d  \hat\omega_k \vert C_{kj}\vert+ \delta^2 \Vert  C\Vert_{1,1}\right\}, \end{equation}
where 
\[\hat\omega_k= \mathbb{E}_{\mathbb{P}_n}\left[\left\vert\bm{x}_k\right\vert\right].\]

The formulation \eqref{adaptive} is a regularized precision matrix.
The term $2\delta\sum_{j,k=1}^d \hat\omega_k \vert C_{kj}\vert$ is scale-adaptive weighted $\ell_1$-regularizer. 
The weights $\omega_k$'s enforce a heterogeneous shrinkage scheme where entries associated with variables of larger scale are penalized more heavily. 
Since high-magnitude features act as high-leverage points for adversarial exploitation, this weighting scheme promotes an estimation that is not only sparse but also intrinsically robust to adversarial manipulation.
The quadratic term $\delta^2 \Vert C\Vert_{1,1}$ addresses the second-order interactions of the adversarial perturbation.
By penalizing the entry-wise magnitude of matrix $C$, this quadratic regularizer serves as a global stability adjustment, ensuring the surrogate remains a rigorous upper bound for the original objective in \eqref{mainproblem2}. 
This whole regularization mechanism  promotes scale-adaptive model selection and adversarial robustness, which we will be illustrated in Section \ref{adversarialrobustnesssec} and Section \ref{varianceadaptive}.

\subsection{Adversarial Robustness}\label{adversarialrobustnesssec}
We demonstrate the adversarial robustness of problem \eqref{adaptive} in the following proposition.
\begin{proposition}[Adversarial Robustness] \label{adversarialrobustness}Suppose $\widetilde{ C}$ is the minimizer of problem \eqref{adaptive}, and $\mathcal{S}(\cdot)$ is the associated objective function. Then, the adversarial robust loss of $\widetilde{ C}$ is upper bounded, i.e., 
    \[   -\log \mathrm{det} \widetilde{ C} +\mathbb{E}_{\mathbb{P}_n}\left[\max_{\Vert\bm{\Delta}\Vert_p\leq \delta}(\bm{x}+\bm{\Delta})^\top \widetilde{ C} (\bm{x}+\bm{\Delta})\right]  \leq \inf_{ C\succ 0} \mathcal{S}( C),\quad 1\leq p\leq \infty.\]
\end{proposition}

Proposition \ref{adversarialrobustness} implies that the adversarially robust loss of the estimation $\widetilde{C}$ obtained from problem \eqref{adaptive} under $\ell_p$-perturbation, which is defined as the worst-case loss under adversarial perturbations posed to the samples, is upper bounded. This property is due to the nature of adversarial training, the motivation behind our proposed problem, and the way we formulate problem \eqref{adaptive} as an upper-bound surrogate.
\subsection{High-dimensional Consistency}\label{varianceadaptive}
This subsection will investigate the high-dimensional convergence rate and model selection consistency for the estimation obtained from  \eqref{adaptive}.
The analysis is based on the primal-dual witness technique innovated  by \cite{wainwright2009sharp}. Our analysis extends the primal-dual witness technique to scenarios involving data-dependent and non-uniform penalty parameters.

We first reformulate problem \eqref{adaptive} equivalently as follows
\begin{equation}\label{highdimformulation}\widetilde{C}\in\arg\inf_{ C\succ 0}\left\{-\log \mathrm{det}  C + \mathbb{E}_{\mathbb{P}_n}[\bm{x}^\top C\bm{x}]+  \sum_{k\neq j} \hat{\lambda}_{kj} \left\vert C_{kj}\right\vert\right\},\end{equation}
\[\hat{\lambda}_{kj}= \delta \mathbb{E}_{\mathbb{P}_n}\left[\left\vert\bm{x}_k\right\vert\right] +\delta \mathbb{E}_{\mathbb{P}_n}\left[\left\vert\bm{x}_j\right\vert\right] + \delta^2,\]
where we do not penalize the diagonals to simplify the analysis.

We then present several notations and necessary assumptions for our subsequent analysis.

Let the off-diagonal support index set of the ground-truth precision matrix $\Sigma^{-1}$ be $E:=\{(i,j): i\not=j,[\Sigma^{-1}]_{ij}\neq 0, i,j=1,\dots,d\},$
and the augmented support index set with diagonals included is denoted by
$S:=E\cup\{(i,i): i=1,\dots,d\}.$ Define the maximum number of nonzeros per row in $\Sigma^{-1}$ with diagonals counted as $s:=\max_{1\le i\le d}\#\{j:\ [\Sigma^{-1}]_{ij}\neq 0, j=1,\dots,d\}.$
Let $\omega_i^\ast:=\mathbb{E}\left[\left\vert\bm{x}_i\right\vert\right]=\sqrt{2{(\pi)^{-1}}\Sigma_{ii}}$ and define $\lambda^\ast_{ij}:=\delta(\omega_i^\ast+\omega_j^\ast)+\delta^2,  i\neq j, \lambda^\ast_{ii}:=0, i,j=1,\dots,d$. 
Let $\Lambda^\ast$ be the symmetric matrix with entries $\lambda^\ast_{ij}$.
Let $\sigma_{\max}:=\max_{1\leq i\leq d}\sqrt{\Sigma_{ii}}$,  $\sigma_{\min}:=\min_{1\leq i\leq d}\sqrt{\Sigma_{ii}}$.
Define
$
\omega_{\max}^\ast:=\max_{1\leq i\leq d}\omega_i^\ast= \sqrt{2{(\pi)^{-1}}}\sigma_{\max},$
$\omega_{\min}^\ast:=\min_{1\leq i\leq d}\omega_i^\ast= \sqrt{2{(\pi)^{-1}}}\sigma_{\min}.
$
We also let Hessian take the form $\Gamma=\Sigma\otimes\Sigma$.

\begin{assumption}[Covariance Control]\label{covariancecontrol}
\[
\kappa_\Gamma:=\Vert (\Gamma_{SS})^{-1}\Vert_\infty<\infty,
\qquad
\kappa_\Sigma:=\Vert\Sigma\Vert_\infty<\infty.
\]
\end{assumption}

\begin{assumption}[Scale-adaptive Incoherence Condition]\label{boundedincoherence} Define $A :=\Gamma_{S^cS} (\Gamma_{SS})^{-1},$
and
    \begin{equation}\label{adaptiveincoherencecondition}\kappa_{A}:=\Vert  A\Vert_\infty <\infty,\quad \psi^\ast
:=\max_{e\in S^c}\sum_{s\in S}|A_{es}|
\frac{\lambda^\ast_s}{\lambda^\ast_e}\leq 1-\alpha,\quad\alpha\in(0,1).\end{equation}
\end{assumption}

\begin{theorem}[Convergence Rate and Model Selection Consistency]
\label{thm:main}

Suppose Assumptions~\ref{covariancecontrol}--\ref{boundedincoherence} hold.
Fix any $\tau>2$ and set
\[
\delta=c_\delta\sqrt{\frac{\log d}{n}},
\qquad 
c_\delta = \frac{\sqrt{\tau}}{\alpha}\max\left\{
16 C_2\sigma_{\max},
\ \ \frac{2\sqrt{2\pi}(1+\kappa_A)C_1\sigma_{\max}^{ 2}}{\sigma_{\min}}
\right\},
\]
where  $C_1, C_2$ are some universal constants.
Define 
the constant
$B:=C_1\sigma_{\max}^{2}\sqrt{\tau} + 3\sqrt{2{(\pi)^{-1}}}\sigma_{\max} c_\delta + c_\delta^{2}.
$
Assume the sample size satisfies
\[
n  \geq \max\left\{ 2\pi C_2^2\left(\frac{\sigma_{\max}}{\sigma_{\min}}\right)^2\tau , 36\kappa_\Gamma^4\kappa_\Sigma^6B^2s^2,\left(
\frac{12\sqrt{2\pi}(1+\kappa_A)\kappa_\Sigma^3 \kappa_\Gamma^2B^2}{\alpha\sigma_{\min}  c_\delta}
\right)^{2}s^2
\right\}\log d.
\]

Then, for the solution $\widetilde{C}$ to problem \eqref{highdimformulation},  with probability at least $\ 1-4d^{2-\tau}$,
the following holds:

\[
\|\widetilde C-\Sigma^{-1}\|_{\max}\ \le\ 2 \kappa_\Gamma B \sqrt{\frac{\log d}{n}},\quad
\widetilde C_{S^c}=0.
\]

 If  the inequality
\[ \min_{(i,j)\in E}|[\Sigma^{-1}]_{ij}|
\ \ge\ 4 \kappa_\Gamma B \sqrt{\frac{\log d}{n}}\]
holds, then $\mathrm{sign}(\widetilde C_{ij})=\mathrm{sign}([\Sigma^{-1}]_{ij})$ for all $(i,j)\in E$.

\end{theorem}

Theorem \ref{thm:main} establishes an error rate of order $\sqrt{\log d/n}$, support recovery, and sign consistency. While these rates align with classic results for the $\ell_1$-regularized estimation \citep{ravikumar2011high}, the conditions required to achieve them are different.

In the classic $\ell_1$-regularized precision matrix estimation, a uniform penalty is applied to all entries. 
The associated model selection consistency is  established under the  incoherence condition: $$\mu^\ast := \max_{e\in S^c} \sum_{s\in S} |A_{es}| \leq 1 - \alpha,\quad \alpha\in(0,1).$$ This condition requires that no null edge is highly correlated with the set of true edges. 
Our perturbed estimation replaces this uniform requirement with the scale-adaptive incoherence condition, as shown in Assumption \ref{boundedincoherence}.
The distinction lies in the penalty ratio $\lambda^*_s/\lambda^*_e$, which essentially compares the marginal standard deviation of the variables involved in a true edge versus a null edge. 

We claim that  the scale-adaptive incoherence condition embedded in our perturbed estimation \eqref{highdimformulation} is superior over the classic one in the $\ell_1$-regularized estimation under heteroscedastic settings. This can be seen from the following proposition.

%==========================================================
% Proposition: heteroscedastic scale inflation vs mitigation
%==========================================================

\begin{proposition}
\label{prop:scale-inflation-mitigation}
Suppose the covariance matrix $\Sigma$ admits the decomposition
$\Sigma = D R D$, where $D=\mathrm{diag}(\sqrt{\Sigma_{11}},\dots,\sqrt{\Sigma_{dd}})$. 
Let $\Gamma^{\dagger}:=R\otimes R$. Define
\[
A^{\dagger} :=  \Gamma_{S^cS}^{\dagger}(\Gamma_{SS}^{\dagger})^{-1}.
\]
If $10\delta<\sqrt{2(\pi)^{-1}}\sigma_{\min}$, then
\[
\mu^\ast=\max_{e\in S^c} \sum_{s\in S} |A_{es}|\leq 
\Big(\frac{\sigma_{\max}}{\sigma_{\min}}\Big)^2\| A^{\dagger}\|_\infty,\quad
\psi^\ast=\max_{e\in S^c}\sum_{s\in S}|A_{es}|
\frac{\lambda^\ast_s}{\lambda^\ast_e}\le\frac{21}{20}\Big(\frac{\sigma_{\max}}{\sigma_{\min}}\Big)\| A^{\dagger}\|_\infty.
\]
\end{proposition}

Compared to the standard incoherence $\mu^\ast$, the scale-adaptive version $\psi^\ast$ reduces the incoherence upper bound by a factor of $\sigma_{\max}(\sigma_{\min})^{-1}$, indicating a possible relaxation of the incoherence condition under the heteroscedastic settings. The improvement will be more significant as the difference between the maximum marginal standard deviation and the minimum one becomes larger.

\section{Regularization Effect}\label{regeffect}
We discuss the regularization effect induced by the adversarial perturbation in the proposed framework.

We characterize this regularization effect in the following proposition.
 \begin{proposition}[Regularization Effect]\label{regularizationprop} Suppose the random variable $\bm{x}$ follows a distribution $\bm{P}$ with $\mathbb{E}_{\bm{P}}[\Vert\bm{x}\Vert_p]<\infty$ and $\bm{P}([ C\bm {x}]_i=0)=0,i=1,\ldots,d$,  $ C$ is a positive definite matrix with $\Vert  C\Vert_{p\to q}<\infty$, and $p\in(1,\infty]$ with  $1/p+1/q=1$. Then, we have that 
\begin{equation*}\mathbb{E}_{\bm{P}}\left[\max_{\Vert \bm{\Delta}\Vert_p\leq \delta}(\bm{x}+\bm{\Delta})^\top  C (\bm{x}+\bm{\Delta})\right]= \mathbb{E}_{\bm{P}} [\bm{x}^\top  C \bm{x}] + 2\delta\mathbb{E}_{\bm{P}} \left[ \Vert  C\bm {x}\Vert_{q}\right]+ \delta^2 \mathbb{E}_{\bm{P}} \left[ \bm{v}_{ C\bm {x}}^\top  C \bm{v}_{ C\bm {x}}\right]+o(\delta^2),\end{equation*}
as $\delta\to 0$, where 
\begin{equation}\label{pairing}\bm{v}_{ C\bm {x}}=\frac{\mathrm{sign}( C\bm {x}) \odot \vert  C\bm {x}\vert^{q-1}}{\Vert  C\bm {x}\Vert_q^{q-1}}.\end{equation}
\end{proposition}

\begin{remark}Since the term $-\log \det C$ in the proposed framework \eqref{intro} is deterministic, Proposition \ref{regularizationprop} focuses on the inner expectation in \eqref{intro} to analyze the induced regularization under $\ell_p$-perturbation.  Instead of only analyzing on the empirical distribution, we consider a more general distribution $\bm{P}$ satisfying certain regularity conditions. The stated regularity conditions in Proposition \ref{regularizationprop} ensure well-posedness and integrability.\end{remark}
\begin{example}\label{example}
We provide the examples for the cases when $p=2$ and $p=\infty$ in Proposition \ref{regularizationprop}: 
\begin{equation*}\begin{aligned}&\mathbb{E}_{\bm{P}}\left[\max_{\Vert \bm{\Delta}\Vert_\infty\leq \delta}(\bm{x}+\bm{\Delta})^\top  C (\bm{x}+\bm{\Delta})\right]\\
=&\mathbb{E}_{\bm{P}} [\bm{x}^\top  C \bm{x}] + 2\delta\mathbb{E}_{\bm{P}} \left[ \Vert  C\bm {x}\Vert_{1}\right]+ \delta^2 \mathbb{E}_{\bm{P}} \left[ \mathrm{sign}( C\bm {x})^\top  C \mathrm{sign}( C\bm {x})\right]+o(\delta^2),\end{aligned}\end{equation*}

\begin{equation*}\begin{aligned}&\mathbb{E}_{\bm{P}}\left[\max_{\Vert \bm{\Delta}\Vert_2\leq \delta}(\bm{x}+\bm{\Delta})^\top  C (\bm{x}+\bm{\Delta})\right]\\
=&\mathbb{E}_{\bm{P}} [\bm{x}^\top  C \bm{x}] + 2\delta\mathbb{E}_{\bm{P}} \left[ \Vert  C\bm {x}\Vert_{2}\right]+ \delta^2 \mathbb{E}_{\bm{P}} \left[ \frac{\bm{x}  C^3\bm{x}}{\bm{x}  C^2\bm{x}}\right]+o(\delta^2),\end{aligned}\end{equation*}
as $\delta\to 0$.
\end{example}

Proposition \ref{regularizationprop} reveals that the proposed adversarially robust precision matrix framework is equivalent to the regularized precision matrix estimation in the asymptotic sense. For small perturbations, the worst-case loss inflates by a first-order term $2\delta\mathbb{E}_{\bm{P}} [ \Vert C\bm {x}\Vert_{q}]$, a second-order term $\delta^2 \mathbb{E}_{\bm{P}} [ \bm{v}_{C\bm {x}}^\top C \bm{v}_{C\bm {x}}]$, and a higher-order residual.

The first-order regularization term arises from Hölder's inequality and penalizes the $\ell_q$ norm of the gradient direction, where $1/p+1/q=1$. The vector $\bm{v}_{C\bm {x}}$ in the second-order term is the $\ell_p$-unit vector that achieves the associated Hölder dual pairing, as defined in \eqref{pairing}.  The second-order regularization term reflects the local curvature of the worst-case loss. Specifically, this term corresponds to the mean of the Rayleigh quotient of $C$ evaluated with respect to the vector $C\bm {x}$ (when $p=2$) or $\mathrm{sign}(C\bm {x})$ (when $p=\infty$). If $C$ has the eigen-decomposition $C = \sum_{j} \lambda_j \bm{u}_j \bm{u}_j^\top$, the associated Rayleigh quotient represents a weighted average of the eigenvalues. For $p=2$, it takes the form: \begin{equation}\label{weight}\frac{\bm{x}^\top C^3\bm{x}}{\bm{x}^\top C^2\bm{x}} =\sum_j \lambda_j \left( \frac{\lambda_j^2 (\bm{u}_j^\top \bm{x})^2}{\sum_k \lambda_k^2 (\bm{u}_k^\top \bm{x})^2} \right).\end{equation}
Similarly, for $p=\infty$:$$\mathrm{sign}(C\bm {x})^\top C\mathrm{sign}(C\bm {x}) = \sum_j \lambda_j  (\bm{u}_j^\top \mathrm{sign}(C\bm {x}))^2.$$Thus, the second-order regularization implies that adversarial perturbation induces a shrinkage effect on the eigenvalues of $C$. Furthermore, as seen in \eqref{weight} for the $p=2$ case, this shrinkage effect is stronger on larger eigenvalues due to the $\lambda^2$ weighting.

We could establish the asymptotic distribution of the adversarially perturbed precision estimation utilizing the regularization effect in  Proposition \ref{regularizationprop}.

\section{Asymptotic Distribution}\label{asynsec}
This section will provide the asymptotic distribution of the adversarially perturbed estimator under different settings.

Recall that the proposed adversarially perturbed estimation is the solution to the following optimization problem:
\begin{equation}\label{section5form} \widehat{ C}\in \arg\inf_{ C\succ 0}\left\{ -\log \mathrm{det}  C +\mathbb{E}_{\mathbb{P}_n} \left[\max_{\Vert \bm{\Delta}\Vert_p \leq \delta}(\bm{x}+\bm{\Delta})^\top  C (\bm{x}+\bm{\Delta})\right]\right\}.\end{equation}

Regarding the perturbation magnitude, we consider a regime where the perturbation level $\delta$ is small and decays as the sample size $n$ increases. Formally, we let $\delta = \delta_n=\eta n^{-\gamma}$ with constants $\eta,\gamma>0$.

The asymptotic distribution of the perturbed estimation $\widehat{C}$ admits an explicit expression as detailed below.
\begin{theorem}[Asymptotic Distribution]\label{asymptoticresults1}
For problem \eqref{section5form}, suppose $p\in(1,\infty]$, and $\delta_n=\eta n^{-\gamma}, \gamma,\eta>0$, then we have that 
\begin{itemize}
    \item  If $0<\gamma<1/2$,
    \[ n^{\gamma}\left(\widehat{ C}-\Sigma^{-1} \right)\Rightarrow -2\eta\Sigma^{-1}  
\mathbb{E}\left[\frac{\left( \mathrm{sign}(\Sigma^{-1}\bm{x}) \odot |\Sigma^{-1}\bm{x}|^{q-1} \right) \bm{x}^\top }{\Vert\Sigma^{-1}\bm{x}\Vert_q^{q-1}}\right]\Sigma^{-1};\]
    \item If $\gamma=1/2$, 
    \[ \sqrt{n}\left(\widehat{ C}-\Sigma^{-1} \right)\Rightarrow  \Sigma^{-1}  
\left(\mathbf{G} -2\eta\mathbb{E}\left[\frac{\left( \mathrm{sign}(\Sigma^{-1}\bm{x}) \odot |\Sigma^{-1}\bm{x}|^{q-1} \right) \bm{x}^\top }{\Vert\Sigma^{-1}\bm{x}\Vert_q^{q-1}}\right]\right) \Sigma^{-1};\]
    
    \item If $\gamma> 1/2$,
     \[ \sqrt{n}\left(\widehat{ C}-\Sigma^{-1} \right)\Rightarrow  \Sigma^{-1}  
\mathbf{G}   \Sigma^{-1},\]    
\end{itemize}
where  $\bm{x}\sim \mathcal{N}(\bm{0},\Sigma)$, $\mathrm{vec} (\mathbf{G})\sim \mathcal{N}(\bm{0}, \Lambda)$, and $\Lambda$ is such that 
\[\mathrm{cov}([\mathbf{G}]_{ij},[\mathbf{G}]_{i^\prime j^\prime})=\mathrm{cov}(\bm{x}_i\bm{x}_j, \bm{x}_{i^\prime}\bm{x}_{j^\prime}).\]
\end{theorem}

The results in Theorem \ref{asymptoticresults1} have some implications.

Firstly,  the proposed perturbed estimation $\widehat{C}$ is a consistent estimator of the precision matrix $\Sigma^{-1}$, i.e., $\widehat{ C}\to_p \Sigma^{-1}$,  when the perturbation magnitude shrinks as the sample size increases.
Then, the proposed perturbed estimation has a convergence order of $\sqrt{n}$ when $\gamma\geq 1/2$ and a slower convergence order of $n^{\gamma}$ when $0<\gamma<1/2$.

Secondly, when we focus on cases where $\gamma\geq 1/2$, the perturbed estimation converges in distribution to a normal distribution.

Thirdly, the perturbation magnitude decays fast enough that the bias induced by the adversarial perturbation becomes asymptotically negligible when $\gamma > 1/2$. In this regime, the estimation $\widehat{C}$ converges to a zero-mean normal distribution identical to that of the unperturbed inverse sample covariance, i.e., the maximum likelihood estimation from problem \eqref{orig}.
This establishes a theoretical threshold: while adversarial training typically introduces a trade-off between robustness and accuracy \citep{tsipras2019robustness}, accuracy and perturbation are not in conflict when the perturbation order exceeds $1/2$ in our proposed adversarially perturbed precision matrix estimation framework.

Recall that we focus on the surrogate objective function instead when we discuss the $\ell_\infty$-perturbed estimation in Section \ref{ellinf}, as shown in the following:
\begin{equation}\label{sec5surr}\widetilde{C}\in\arg\inf_{ C\succ 0}\left\{-\log \mathrm{det}  C + \mathbb{E}_{\mathbb{P}_n}\left[\bm{x}^\top  C\bm{x}\right]+ 2\delta\sum_{j=1}^d \sum_{k=1}^d  \omega_k \vert C_{kj}\vert+ \delta^2 \Vert  C\Vert_{1,1}\right\}. \end{equation}

In contrast to the asymptotic results in Theorem \ref{asymptoticdistributionsurrogate}, the asymptotic distribution of the estimation $\widetilde{C}$ in problem \eqref{sec5surr} does not admit a simple closed form but can be characterized as the minimizer of a stochastic process, as shown in the following theorem.

\begin{theorem}[Asymptotic Distribution when $p=\infty$ with Surrogate Objective] \label{asymptoticdistributionsurrogate}
    If $\delta_n=\eta n^{-\gamma}, \eta,\gamma>0$, for the estimation $\widetilde{ C}$ obtained from problem \eqref{sec5surr}, we have that 
\[\sqrt{n}\left(\widetilde{ C}-\Sigma^{-1} \right)\Rightarrow \arg\min_{\bm{U}} W(\bm{U}).\]

\begin{itemize}
    \item  If $0<\gamma<1/2$,
\begin{equation*}\begin{aligned} W(\bm{U}) =& \frac{1}{2}\mathrm{tr}(\bm{U}\Sigma \bm{U}\Sigma)\\&+  \sqrt{\frac{8}{\pi}}\eta  \sum_{j=1}^d \sum_{k=1}^d    \Sigma_{kk} \left( \bm{U}_{jk}\mathrm{sign}([\Sigma^{-1}]_{jk})\mathbb{I}([\Sigma^{-1}]_{jk}\not = 0)+\vert  \bm{U}_{jk} \vert \mathbb{I}([\Sigma^{-1}]_{jk} = 0)\right).\end{aligned}\end{equation*}
    \item If $\gamma=1/2$, 
\begin{equation*}\begin{aligned} W(\bm{U}) =& -\mathrm{tr}(\bm{U}\mathbf{G}) + \frac{1}{2}\mathrm{tr}(\bm{U}\Sigma \bm{U}\Sigma)\\&+ \sqrt{\frac{8}{\pi}}\eta  \sum_{j=1}^d \sum_{k=1}^d    \Sigma_{kk} \left( \bm{U}_{jk}\mathrm{sign}([\Sigma^{-1}]_{jk})\mathbb{I}([\Sigma^{-1}]_{jk}\not = 0)+\vert  \bm{U}_{jk} \vert \mathbb{I}([\Sigma^{-1}]_{jk} = 0)\right).\end{aligned}\end{equation*}
    
    \item If $\gamma> 1/2$,
\[W(\bm{U})  = -\mathrm{tr}(\bm{U}\mathbf{G}) +\frac{1}{2}\mathrm{tr}(\bm{U}\Sigma \bm{U}\Sigma),\]    
i.e.,
     \[ \sqrt{n}\left(\widetilde{ C}-\Sigma^{-1} \right)\Rightarrow  \Sigma^{-1}  
\mathbf{G}   \Sigma^{-1},\]   
\end{itemize}
where  $\mathrm{vec} (\mathbf{G})\sim \mathcal{N}(\bm{0}, \Lambda)$, and $\Lambda$ is such that 
\[\mathrm{cov}([\mathbf{G}]_{ij},[\mathbf{G}]_{i^\prime j^\prime})=\mathrm{cov}(\bm{x}_i\bm{x}_j, \bm{x}_{i^\prime}\bm{x}_{j^\prime}).\]
\end{theorem}
\begin{corollary}[Sparsity Recovery]\label{sparsityrecovery}
Let $\widetilde{C}$ be the minimizer of problem \eqref{adaptive}. Suppose that $\delta = \eta n^{-1/2}$ with $\eta > 0$. Then, as $n \to \infty$, the asymptotic distribution of $\widetilde{C}_{ij}$ has a positive probability mass at $0$ whenever $[\Sigma^{-1}]_{ij} = 0$.
\end{corollary}
\begin{remark}
Corollary \ref{sparsityrecovery} follows immediately from Theorem \ref{asymptoticdistributionsurrogate} and Proposition 3.8 in \cite{xie2025asymptotic}.
\end{remark}

Theorem \ref{asymptoticdistributionsurrogate} shows that the estimation $\widetilde{C}$ in problem \eqref{sec5surr} shares the same asymptotic normal distribution of the estimation $\widehat{C}$ in problem \eqref{section5form}.
However, for $\gamma = 1/2$, the asymptotic distribution of $\widetilde{C}$ is non-normal. The term $\Sigma_{kk}$ presented in Theorem \ref{asymptoticdistributionsurrogate} is resulted from the adversarial perturbation in the framework \eqref{sec5surr}, showing the asymptotic effect of the scale-adaptivity and adversarial robustness discussed in Section \ref{adversarialrobustnesssec} and Section \ref{varianceadaptive}.

In addition, our asymptotic analysis demonstrates that the sparsity recovery can  be achievable when $\gamma = 1/2$, seeing Corollary \ref{sparsityrecovery}. 
Recall the high dimensional finite-sample analysis implies that model selection consistency is obtained when $\delta$ is in the order of $n^{-1/2}$. 
This provides a unification of finite-sample and asymptotic perspectives for $\ell_\infty$-perturbed estimation, demonstrating that both regimes demand an equivalent scaling of $\delta \asymp n^{-1/2}$ to guarantee model selection consistency.

To have a comprehensive understanding and comparison, we conclude some of the key characteristics of the asymptotic distribution of the perturbed estimation in Table \ref{tab}.

\begin{table}[H]
    \centering
    \begin{tabular}{c|c|c|c}
       \hline
        &$p\in(1,\infty]$, problem \eqref{intro}   & $p=\infty$, surrogate, problem \eqref{adaptive} & Rate \\
        \hline
        $\gamma\in(0,1/2)$& constant    & constant & $n^{\gamma}$ \\
        \hline
        $\gamma=1/2$ &  normality with bias   & non-normality with bias, sparsity & $\sqrt{n}$\\
        \hline
        $\gamma>1/2$ &  \multicolumn{2}{c|}{normality without bias} & $\sqrt{n}$ \\
        \hline
    \end{tabular}
    \caption{Asymptotic Distribution of the Perturbed Estimation}
    \label{tab}
\end{table}

\section{Numerical Experiments}\label{numsection}
In this section, we will deliver numerical experiments for the proposed adversarially perturbed precision matrix estimation scheme and compare it with the $\ell_1$-regularized precision matrix estimation \citep{friedman2008sparse}. The experiments are implemented by R. The code url:\url{https://github.com/YilingXie27/Perturbed-Precision-Matrix-Estimation}.

We will focus on the surrogate $\ell_\infty$-perturbed matrix estimation framework  \eqref{adaptive} discussed in Section \ref{ellinf}.
While the $\ell_\infty$-perturbed precision matrix estimation has penalties dependent on the observed samples, the $\ell_1$-regularized precision matrix estimation \citep{friedman2008sparse} has a constant penalty, as shown in the following:
\begin{equation}\label{standard}
\inf_{ C\succ 0}\left\{-\log \mathrm{det}  C + \mathrm{tr}(\bar{A} C)+ \lambda\Vert  C\Vert_{1}\right\}. \end{equation}

In practice, we do not penalize the diagonals for either method to promote more desirable numerical performance for all experiments in this section.

We may use the following metrics to evaluate the numerical performance :
true negative rate (TNR),
true positive rate (TPR), 
accuracy (ACC),
and Matthews Correlation Coefficient (MCC) \citep{guilford1954psychometric}. The metrics are defined as follows:
\[\mathrm{TNR} = \frac{\mathrm{TN}}{\mathrm{TN}+\mathrm{FP}},\quad \mathrm{TPR} = \frac{\mathrm{TP}}{\mathrm{TP}+\mathrm{FN}},\]
\[\mathrm{ACC}=\frac{\mathrm{TN} + \mathrm{TP}}{\mathrm{TN}+\mathrm{FP}+\mathrm{TP}+\mathrm{FN}},\]
%\[,\quad \mathrm{TPR} = \frac{\mathrm{TP}}{\mathrm{TP}+\mathrm{FN}},\]
\[\mathrm{MCC} = \frac{\mathrm{TP} \times \mathrm{TN} - \mathrm{FP} \times \mathrm{FN}}{\sqrt{(\mathrm{TP} + \mathrm{FP})(\mathrm{TP} + \mathrm{FN})(\mathrm{TN} + \text{FP})(\text{TN} + \text{FN})}},\]
where
$\mathrm{TP}$, $\mathrm{TN}$, $\mathrm{FP}$,  $\mathrm{FN}$ denote the number of true positives, true negatives, false positives, and
false negatives, respectively.
% In addition, the KL loss is defined as:
% \[ \mathrm{KL} = -\log\vert\bar{C}\vert + \mathrm{tr}(\bar{C}\Sigma) - \left( -\log\vert\Sigma^{-1}\vert + d \right),\]
% where
% $\Sigma$ is the ground-truth covariance matrix, and $\bar{C}$ is the estimation of the precision matrix.
\subsection{Experiments with Synthetic Data}
 We will carry out the numerical experiments on the following synthetic models.

\noindent \textit{Model 1.} An $AR(2)$ model with $c_{i,i} = 1$, $c_{i,i-1} = c_{i-1,i} = 0.5$ and $c_{i,i-2} = c_{i-2,i} = 0.25$, for $i=1,\ldots,d$.
 
\noindent \textit{Model 2.} An $AR(3)$ model with $c_{i,i} = 1$, $c_{i,i-1} = c_{i-1,i} = 0.4$ and $c_{i,i-2} = c_{i-2,i} = c_{i,i-3} = c_{i-3,i} = 0.2,$ for $i=1,\ldots,d$.
 
\noindent \textit{Model 3.} An $AR(4)$ model with $c_{i,i} = 1$, $c_{i,i-1} = c_{i-1,i} = 0.4$, $c_{i,i-2} = c_{i-2,i} = c_{i,i-3} = c_{i-3,i} = 0.2$ and $c_{i,i-4} = c_{i-4,i} = 0.1,$ for $i=1,\ldots,d$.
 
%\noindent \textit{Model 6.}\quad Full model with $c_{i,j} = 2$ if $i = j$ and $c_{i,j} = 1$ otherwise, for $i,j=1,\ldots,d$.
 
\noindent \textit{Model 4.} Star model with every node connected to the first node, with $c_{i,i} = 1$, $c_{1,i} = c_{i,1} = 0.2$ and $c_{i,j} = 0$ otherwise, for $i,j=1,\ldots,d$.

\noindent \textit{Model 5.} Circle model with $c_{i,i} = 1$, $c_{i,i-1} = c_{i-1,i} = 0.5$ and $c_{1,d} = c_{d,1} = 0.4$,  for $i=1,\ldots,d$.

We inject heteroskedastic marginal variances by applying a diagonal rescaling into the models above to illustrate the superiority of the perturbed method under heterogeneous variance settings, as discussed in Section \ref{varianceadaptive}.
Define a diagonal scaling matrix $S=\mathrm{diag}(s_1,\ldots,s_d)$ with
$s_i=10$ for $i=1,\ldots,5$ and $s_i=1$ for $i=6,\ldots,d$.
We then set the data-generating covariance to
$\Sigma_S= S\Sigma S,$.
Because $D$ is diagonal, the sparsity pattern of the corresponding precision matrix is unchanged, while the first five variables have substantially larger scales.

For each model, we run simulations by the QUIC algorithm \citep{hsieh2014quic} with dimension $d=30$ and sample sizes $n=\{20, 30, 40\}$. 
We consider equally spaced grids for the perturbation magnitude $\delta$  in problem \eqref{adaptive} and the regularization parameter $\lambda$ in problem \eqref{standard} over an equally spaced grid on $[0.01, 1]$. 
The value of $\delta$ and $\lambda$ is chosen based on the Bayesian Information Criterion.
For each setting, we run two methods 100 times and record the average and standard error.

We report the results of the perturbed method and the $\ell_1$-regularized method in Tables \ref{ar2table}–\ref{circletable}. Since our primary interest lies in heteroscedastic settings, we additionally report the performance of the $\ell_1$-regularized method applied to standardized data. Across the majority of scenarios, the perturbed method attains higher ACC and MCC. This observation is consistent with our theoretical results on improved model selection consistency under heteroscedastic models.
\begin{table}[H]
\centering % Adds vertical padding for readability
\begin{tabular}{c|c|cccc}
\hline
  $n$& Method & ACC & MCC & TNR & TPR \\ \hline
\multirow{2}{*}{$20$} & Perturbed &\textbf{0.876} (0.000) &\textbf{0.232} (0.004)  &0.994 (0.000)& 0.095 (0.002)\\  
 & $\ell_1$ & 0.233 (0.001) &0.036 (0.004)&  0.132 (0.001) & \textbf{0.903} (0.004)\\  & $\ell_1$ std & 0.869 (0.000)& 0.016 (0.005)  &\textbf{0.999} (0.000) &0.006 (0.002)\\\hline
\multirow{2}{*}{$30$} & Perturbed & \textbf{0.877} (0.000) &\textbf{0.243} (0.003) & 0.994 (0.000)& 0.102 (0.002)\\ 
 & $\ell_1$ &0.698 (0.002)& -0.002 (0.003)& 0.769 (0.002)& \textbf{0.229} (0.005) \\
 & $\ell_1$ std &0.871 (0.001) &0.051 (0.011) & \textbf{0.998} (0.001) &0.026 (0.007)\\
 \hline
\multirow{2}{*}{$40$} & Perturbed &\textbf{0.879} (0.000)& \textbf{0.263} (0.004)&  \textbf{0.994} (0.001) &0.115 (0.003)\\  
 & $\ell_1$ & 0.709 (0.001)& 0.010 (0.004) & 0.781 (0.002) &\textbf{0.232} (0.005)
 \\  
 & $\ell_1$ std &0.875 (0.001) &0.148 (0.018) & 0.993 (0.002) &0.096 (0.014) \\ \hline
\end{tabular}
\caption{Average Value and Standard Error of the Metrics for $AR(2)$}
\label{ar2table}
\end{table}

\begin{table}[H]
\centering
\begin{tabular}{c|c|cccc}
\hline
$n$ & Method & ACC & MCC & TNR & TPR \\ \hline
\multirow{3}{*}{$20$} 
& Perturbed      & \textbf{0.818} (0.000) & \textbf{0.218} (0.003) & 0.998 (0.000) & 0.069 (0.002) \\
& $\ell_1$       & 0.318 (0.012)          & 0.000 (0.005)          & 0.203 (0.020) & \textbf{0.796} (0.022) \\
& $\ell_1$ std   & 0.807 (0.000)          & 0.005 (0.002)          & \textbf{1.000} (0.000) & 0.001 (0.001) \\
\hline
\multirow{3}{*}{$30$} 
& Perturbed      & \textbf{0.820} (0.000) & \textbf{0.228} (0.003) & 0.998 (0.000) & 0.074 (0.002) \\
& $\ell_1$       & 0.662 (0.001)          & -0.047 (0.002)         & 0.780 (0.002) & \textbf{0.171} (0.002) \\
& $\ell_1$ std   & 0.808 (0.000)          & 0.027 (0.007)          & \textbf{0.999} (0.000) & 0.009 (0.003) \\
\hline
\multirow{3}{*}{$40$} 
& Perturbed      & \textbf{0.819} (0.000) & \textbf{0.223} (0.003) & \textbf{0.998} (0.000) & 0.072 (0.002) \\
& $\ell_1$       & 0.672 (0.001)          & -0.035 (0.003)         & 0.791 (0.002) & \textbf{0.173} (0.003) \\
& $\ell_1$ std   & 0.813 (0.001)          & 0.085 (0.014)          & 0.996 (0.001) & 0.045 (0.009) \\
\hline
\end{tabular}
\caption{Average Value and Standard Error of the Metrics for $AR(3)$}
\label{ar3table}
\end{table}

\begin{table}[H]
\centering
\begin{tabular}{c|c|cccc}
\hline
$n$ & Method & ACC & MCC & TNR & TPR \\ \hline
\multirow{3}{*}{$20$}
& Perturbed    & \textbf{0.761} (0.000) & \textbf{0.203} (0.003) & 0.999 (0.000) & 0.058 (0.001) \\
& $\ell_1$     & 0.407 (0.014)          & -0.029 (0.005)         & 0.324 (0.029) & \textbf{0.650} (0.032) \\
& $\ell_1$ std & 0.747 (0.000)          & 0.005 (0.002)          & \textbf{1.000} (0.000) & 0.001 (0.000) \\
\hline
\multirow{3}{*}{$30$}
& Perturbed    & \textbf{0.761} (0.000) & \textbf{0.199} (0.002) & \textbf{1.000} (0.000) & 0.054 (0.001) \\
& $\ell_1$     & 0.618 (0.001)          & -0.078 (0.002)         & 0.776 (0.002) & \textbf{0.151} (0.001) \\
& $\ell_1$ std & 0.748 (0.000)          & 0.012 (0.004)          & \textbf{1.000} (0.000) & 0.002 (0.001) \\
\hline
\multirow{3}{*}{$40$}
& Perturbed    & \textbf{0.760} (0.000) & \textbf{0.195} (0.002) & \textbf{1.000} (0.000) & 0.053 (0.001) \\
& $\ell_1$     & 0.631 (0.001)          & -0.059 (0.002)         & 0.793 (0.002) & \textbf{0.154} (0.001) \\
& $\ell_1$ std & 0.749 (0.001)          & 0.031 (0.007)          & 0.999 (0.001) & 0.011 (0.004) \\
\hline
\end{tabular}
\caption{Average Value and Standard Error of the Metrics for $AR(4)$}
\label{ar4table}
\end{table}
\begin{table}[H]
\centering
\begin{tabular}{c|c|cccc}
\hline
$n$ & Method & ACC & MCC & TNR & TPR \\ \hline
\multirow{3}{*}{$20$}
& Perturbed    & \textbf{0.866} (0.008) & \textbf{0.534} (0.016) & \textbf{0.861} (0.008) & 0.929 (0.009) \\
& $\ell_1$     & 0.837 (0.009)          & 0.496 (0.006)          & 0.826 (0.010)          & \textbf{0.989} (0.002) \\
& $\ell_1$ std & 0.641 (0.009)          & 0.157 (0.006)          & 0.639 (0.010)          & 0.666 (0.012) \\
\hline
\multirow{3}{*}{$30$}
& Perturbed    & \textbf{0.886} (0.006) & \textbf{0.581} (0.011) & \textbf{0.879} (0.006) & 0.981 (0.005) \\
& $\ell_1$     & 0.861 (0.001)          & 0.524 (0.002)          & 0.851 (0.001)          & \textbf{0.996} (0.001) \\
& $\ell_1$ std & 0.644 (0.003)          & 0.208 (0.004)          & 0.635 (0.003)          & 0.772 (0.006) \\
\hline
\multirow{3}{*}{$40$}
& Perturbed    & \textbf{0.902} (0.004) & \textbf{0.611} (0.007) & \textbf{0.896} (0.004) & 0.997 (0.001) \\
& $\ell_1$     & 0.873 (0.001)          & 0.546 (0.002)          & 0.864 (0.001)          & \textbf{0.999} (0.000) \\
& $\ell_1$ std & 0.628 (0.002)          & 0.236 (0.003)          & 0.612 (0.002)          & 0.854 (0.006) \\
\hline
\end{tabular}
\caption{Average Value and Standard Error of the Metrics for the Star model}
\label{startable}
\end{table}

\begin{table}[H]
\centering
\begin{tabular}{c|c|cccc}
\hline
$n$ & Method & ACC & MCC & TNR & TPR \\ \hline
\multirow{3}{*}{$20$}
& Perturbed    & 0.622 (0.014)          & 0.310 (0.008)          & 0.594 (0.015)          & 0.993 (0.002) \\
& $\ell_1$     & 0.398 (0.015)          & 0.181 (0.007)          & 0.355 (0.016)          & 0.980 (0.002) \\
& $\ell_1$ std & \textbf{0.633} (0.012) & \textbf{0.320} (0.009) & \textbf{0.606} (0.013) & \textbf{0.995} (0.001) \\
\hline
\multirow{3}{*}{$30$}
& Perturbed    & 0.706 (0.004)          & 0.363 (0.004)          & 0.684 (0.005)          & \textbf{1.000} (0.000) \\
& $\ell_1$     & 0.587 (0.005)          & 0.283 (0.003)          & 0.557 (0.005)          & 0.998 (0.001) \\
& $\ell_1$ std & \textbf{0.713} (0.011) & \textbf{0.381} (0.008) & \textbf{0.692} (0.012) & 0.999 (0.000) \\
\hline
\multirow{3}{*}{$40$}
& Perturbed    & \textbf{0.707} (0.003) & \textbf{0.363} (0.003) & \textbf{0.685} (0.004) & \textbf{1.000} (0.000) \\
& $\ell_1$     & 0.606 (0.004)          & 0.294 (0.002)          & 0.577 (0.004)          & \textbf{1.000} (0.000) \\
& $\ell_1$ std & 0.666 (0.011)          & 0.346 (0.009)          & 0.642 (0.012)          & \textbf{1.000} (0.000) \\
\hline
\end{tabular}
\caption{Average Value and Standard Error of the Metrics for the Circle model}
\label{circletable}
\end{table}

\subsection{Experiments with Real Data}
In this subsection, we run the two estimation methods on the leukemia dataset \citep{golub1999molecular} to compare the associated real-world practical performance. 

The dataset comprises 7129 gene expression profiles collected from 72 patients, of whom 47 were diagnosed with Acute Lymphoblastic Leukemia (ALL) and 25 with Acute Myeloid Leukemia (AML).
We first select the top $d$ genes with the highest marginal variance. The dimension $d$ is varied across the set $\{20, 30, 40, 50, 60\}$. The data is then centered and standardized. Following the Linear Discriminant Analysis (LDA) framework,  the classification rule relies on the LDA score. For the observation $\bm{x}$, the LDA score is defined by 
\[
\delta_k(\bm{x}) = \bm{x}^T \hat{ C}\hat{\boldsymbol{\mu}}_k - \frac{1}{2} \hat{\boldsymbol{\mu}}_k^T \hat{ C}\hat{\boldsymbol{\mu}}_k + \log \hat{\pi}_k,\quad k=1,2, \]
where $\hat{\pi}_k$ represents the proportion of group $k$ in the training set and $\hat{\boldsymbol{\mu}}_k$ is the empirical mean of the observations in the training set. The classification rule is $\hat{k}(\bm{x})  =\arg\max_{k \in \{1, 2\}} \delta_k(\bm{x})$.  The classification performance is  dependent on the estimation accuracy of the precision matrix.

The experiment employs a nested cross-validation scheme. In the outer loop, the data is split into 10 stratified folds. For each training run, we additionally carry out an inner stratified 5‑fold cross‑validation to tune the parameters, $\delta$ in the perturbed method \eqref{adaptive} and $\lambda$ in the standard $\ell_1$-regularized method \eqref{standard}, through a grid search within $[0.01,1]$, choosing the value that has the best MCC. With the optimal parameter, we train the model on the full training set and then evaluate the model performance on the held‑out test fold. We replicate this whole scheme $100$ times.

To evaluate performance, we report ACC, MCC, TNR, and TPR. We treat AML as the positive class and ALL as the negative class, thus TPR corresponds to sensitivity for AML and TNR corresponds to specificity for ALL. The average metrics and their standard errors are computed across the 100 replications, where the standard error is calculated from the variability of the replicate-level means.

Notably, across the considered values of $d$, the perturbed method generally achieves higher ACC, MCC, and TPR. Overall, these results suggest that the proposed perturbed estimator can provide improved precision matrix estimates for this dataset, translating into better LDA classification performance compared with the $\ell_1$-regularized approach.

\begin{table}[H]
\centering
\begin{tabular}{c|c|cccc}
\hline
  $d$& Method & ACC & MCC & TNR & TPR \\ \hline

\multirow{2}{*}{$20$} & Perturbed
& \textbf{0.866} (0.002) & \textbf{0.682} (0.005) & \textbf{0.997} (0.001) & \textbf{0.618} (0.005)\\
& $\ell_1$
& 0.863 (0.002) & 0.670 (0.006) & \textbf{0.997} (0.001) & 0.607 (0.006)\\ \hline

\multirow{2}{*}{$30$} & Perturbed
& \textbf{0.921} (0.002) & \textbf{0.822} (0.005) & \textbf{1.000} (0.000) & \textbf{0.771} (0.006)\\
& $\ell_1$
& 0.918 (0.002) & 0.817 (0.004) & \textbf{1.000} (0.000) & 0.764 (0.005)\\ \hline

\multirow{2}{*}{$40$} & Perturbed
& \textbf{0.938} (0.002) & \textbf{0.863} (0.004) & \textbf{0.998} (0.001) & \textbf{0.823} (0.004)\\
& $\ell_1$
& 0.937 (0.002) & 0.860 (0.004) & \textbf{0.998} (0.001) & 0.820 (0.005)\\ \hline

\multirow{2}{*}{$50$} & Perturbed
& \textbf{0.944} (0.001) & \textbf{0.878} (0.003) & 0.997 (0.001) & \textbf{0.844} (0.004)\\
& $\ell_1$
& 0.939 (0.001) & 0.866 (0.004) & \textbf{0.998} (0.001) & 0.826 (0.004)\\ \hline

\multirow{2}{*}{$60$} & Perturbed
& \textbf{0.938} (0.001) & \textbf{0.864} (0.004) & 0.990 (0.001) & \textbf{0.838} (0.004)\\
& $\ell_1$
& 0.935 (0.001) & 0.859 (0.003) & \textbf{0.993} (0.001) & 0.825 (0.004)\\ \hline

\end{tabular}
\caption{Average Value and Standard Error of the Metrics for leukemia dataset}
\label{performance}
\end{table}
% Acknowledgements and Disclosure of Funding should go at the end, before appendices and references
\section{Discussion}\label{dis}

In addition to robustness and model selection consistency, future directions may include extending our framework to handle more complex data structures, including the unknown group structures studied in \citep{cheng2025large} and the compositional data investigated in \citep{zhang2025care}.

While the adversarial perturbations investigated in our framework are sample-wise, as defined in \eqref{intro}, an intriguing direction for future research would be to extend this to feature-wise perturbations.
In the linear regression, \citet{xu2008robust} has shown that
robust linear regression framework under the feature-wise perturbation is equivalent to LASSO.
Inspired by this fact, the computational and statistical explorations of feature-wise perturbed precision matrix estimation is of much interest.


\newpage

\appendix
\section{Proof of Theorem \ref{tractable}}
The proof for Theorem  \ref{tractable} is based on the following lemmas.

\begin{lemma}\label{lemma1}
The dual of the optimization problem 
    \[\max_{\Vert \bm{\Delta}\Vert_2\leq \delta}(\bm{x}+\bm{\Delta})^\top  C (\bm{x}+\bm{\Delta})\]
     is the following optimization problem
     \[\bm{x}^\top  C\bm{x}+\inf_{\lambda I - C \succ 0}\left\{\lambda\delta^2+\bm{x}^\top  C (\lambda I - C )^{-1}  C\bm {x}\right\}.\]
\end{lemma}
\begin{proof}
    The Lagrangian dual of $\max_{\Vert \bm{\Delta}\Vert_2\leq \delta}(\bm{x}+\bm{\Delta})^\top  C (\bm{x}+\bm{\Delta})$ is as follows:
    \[\inf_{\lambda\geq 0}\left\{\lambda\delta^2+\max_{\bm{\Delta} }\left\{(\bm{x}+\bm{\Delta})^\top  C (\bm{x}+\bm{\Delta})-\lambda\Vert \bm{\Delta}\Vert_2^2\right\}\right\},\]
    which is equivalent to
        \[\bm{x}^\top  C\bm{x}+\inf_{\lambda\geq 0}\left\{\lambda\delta^2+\max_{\bm{\Delta}}\left\{\bm{\Delta}^\top (C-\lambda I)\bm{\Delta}+2 \bm{x}^\top  C \bm{\Delta}\right\}\right\}.\]
        
    To make the inner optimization problem finite, we should require $\lambda I - C \succ 0$. In this way, the optimal value will be obtained when 
    \[\bm{\Delta}= (\lambda I - C )^{-1} C\bm {x},\]
    the resulting optimal value is 
    \[\max_{\bm{\Delta}\in\mathbb{R}^p}\left\{\bm{\Delta}^\top (C-\lambda I)\bm{\Delta}+2 \bm{x}^\top  C \bm{\Delta}\right\}= \bm{x}^\top  C (\lambda I - C )^{-1}  C\bm {x}\]
\end{proof}
Now we begin to prove Theorem \ref{tractable}.
\begin{proof}
It follows from Lemma \ref{lemma1} that problem \eqref{mainproblem1} is equivalent to the following:
\[\inf_{ C\succ 0,\lambda I - C \succ 0}\left\{ -\log \mathrm{det}  C+\frac{1}{n}\sum_{i=1}^n\bm{x}_i^\top  C\bm{x}_i+\lambda\delta^2+\frac{1}{n}\sum_{i=1}^n \bm{x}_i^\top  C (\lambda I - C )^{-1}  C\bm {x}_i\right\}.\]

By introducing new variables $t_i$, we have that 
\[\inf_{ C\succ 0, \lambda I - C \succ 0}\left\{ -\log \mathrm{det}  C+\frac{1}{n}\sum_{i=1}^n\bm{x}_i^\top  C\bm{x}_i+\lambda\delta^2+\frac{1}{n}\sum_{i=1}^n t_i\right\}.\]
\[ \mathrm{s.t.}\quad  \bm{x}_i^\top  C (\lambda I - C )^{-1}  C\bm {x}_i\leq t_i\quad \text{for } 1\leq i\leq n.\]

Due to the Schur complement, we have the following reformulation: 
\[
\begin{aligned}
\inf_{\lambda, \{t_i\}, C} \quad &  -\log \mathrm{det}  C+\frac{1}{n}\sum_{i=1}^n\bm{x}_i^\top  C\bm{x}_i+\lambda\delta^2+\frac{1}{n}\sum_{i=1}^n t_i \\
\text{subject to} \quad & C\succ 0,\lambda I - C \succ 0\\
& \begin{pmatrix}
t_i &  \bm{x}_i^\top C \\
C\bm {x}_i & \lambda I - C 
\end{pmatrix} \succeq 0, \quad \forall i = 1, \dots, n.\\
\end{aligned}\]
\end{proof}
\section{Proof of Theorem \ref{equivalence}}
\begin{proof}
It remains to show the following identity: 
\[\mathrm{tr}(C(\lambda I-  C)^{-1} C \bar{A} ) =\lambda^2 \mathrm{tr}((\lambda I-  C)^{-1} \bar{A})-\mathrm{tr}((\lambda I+ C)\bar{A}).\]

Since a matrix can commute with its own resolvent matrix,  we have that 
\[\mathrm{tr}(C(\lambda I-  C)^{-1} C \bar{A} ) = \mathrm{tr}(C^2(\lambda I-  C)^{-1}\bar{A} ).\]

Then, we can conclude that 
\[\mathrm{tr}(C(\lambda I-  C)^{-1} C \bar{A} ) -\lambda^2 \mathrm{tr}((\lambda I-  C)^{-1} \bar{A})= \mathrm{tr}((C^2-\lambda^2 I)(\lambda I-  C)^{-1}\bar{A} ).\]

Since 
\[(C^2-\lambda^2 I) = (C-\lambda I)(C+\lambda I),\] the identity has been proven.
\end{proof}

\section{Proof of Proposition \ref{adversarialrobustness}}
\begin{proof}
When $p\geq 1$, we have that 
\begin{equation*}\begin{aligned}
\max_{\Vert \bm{\Delta}\Vert_p\leq \delta}(\bm{x}+\bm{\Delta})^\top  C (\bm{x}+\bm{\Delta})&\leq 
\max_{\Vert \bm{\Delta}\Vert_\infty\leq \delta}(\bm{x}+\bm{\Delta})^\top  C (\bm{x}+\bm{\Delta})\\& = \bm{x} ^\top  C  \bm{x} +\max_{\Vert \bm{\Delta}\Vert_\infty\leq \delta} \left\{ 2\bm{x}^\top  C \bm{\Delta} +  \bm{\Delta} ^\top  C  \bm{\Delta}\right\}\\&\leq \bm{x} ^\top  C  \bm{x} +\max_{\Vert \bm{\Delta}\Vert_\infty\leq \delta}   2\bm{x}^\top  C \bm{\Delta} +\max_{\Vert \bm{\Delta}\Vert_\infty\leq \delta}  \bm{\Delta} ^\top  C  \bm{\Delta}\\&= \bm{x} ^\top  C  \bm{x} +2\delta \Vert  C\bm {x}\Vert_1+ \max_{\Vert \bm{\Delta}\Vert_\infty\leq \delta}  \bm{\Delta} ^\top  C  \bm{\Delta}\\&\leq  \bm{x} ^\top  C  \bm{x} +2\delta \Vert  C\bm {x}\Vert_1+\delta^2 \Vert  C\Vert_{1,1}
\end{aligned}\end{equation*}

Notice that 
\begin{equation*}\begin{aligned}\mathbb{E}_{\mathbb{P}_n}\left[ \Vert  C\bm {x}\Vert_1\right]=\frac{1}{n}\sum_{i=1}^n \sum_{j=1}^d \vert [ C]_j\bm{x}_i\vert &\leq \frac{1}{n}\sum_{i=1}^n \sum_{j=1}^d \sum_{k=1}^d \vert [ C]_{kj}\vert \vert\bm{x}^i_k\vert \\
&=  \sum_{j=1}^d \sum_{k=1}^d \vert [ C]_{kj}\vert \left(\frac{1}{n}\sum_{i=1}^n \vert\bm{x}^i_k\vert\right).\end{aligned}\end{equation*}

We can conclude that the inequality holds.
\end{proof}
\section{Proof for Theorem \ref{thm:main}}
Our proof for Theorem \ref{thm:main} is based on the primal-dual witness (PDW) construction proposed by \cite{wainwright2009sharp,ravikumar2011high}.

Define the symmetric matrix $\hat\Lambda$ by $\hat\Lambda_{jj}=0$ and
$\hat\Lambda_{kj}=\hat\lambda_{kj}$ for $k\neq j, k,j=1,\ldots,d$.

The matrix $\widetilde{C}\succ 0$ is optimal for \eqref{highdimformulation} if and only if there exists a symmetric $\widetilde Z$ such that
\begin{equation}
\label{eq:KKT_thm}
\bar{A} - \widetilde{C}^{-1} + \hat\Lambda\odot \widetilde Z=0,
\end{equation}
where $\widetilde  Z\in \partial \Vert \widetilde{C}\Vert_{1,1}$, $\bar{A}=\mathbb{E}_{\mathbb{P}_n}[\bm{x}\bm{x}^\top]$.

Define the restricted problem
\begin{equation}
\label{eq:restricted_thm}
\dot{C}
\in\arg\min_{C\succ 0,\ C_{S^c}=0}
\Big\{-\log\det C + \langle\bar{A},C\rangle + \sum_{k\neq j}\hat\lambda_{kj}|C_{kj}|\Big\},
\end{equation}
and set $\dot Z_S=\mathrm{sign}(\dot C_S)$. 

For $(j,k)\in S^{c}$, we  define
\begin{equation}
\label{eq:Zoff_def_thm}
\dot Z_{jk}
:=
\frac{-\bar{A}_{jk}+{[{\dot C}^{-1}}]_{jk}}{\hat\lambda_{jk}} .
\end{equation}

Let ${\it \Delta}:= \dot{C}-\Sigma^{-1}$, $W: = \bar{A}-\Sigma$, and
\[
R({\it\Delta}) :=\dot{C}^{-1} - \Sigma+\Sigma{\it \Delta}\Sigma.
\]

We present several lemmas before further proceeding to prove Theorem \ref{thm:main}.
\begin{lemma}[Remainder Bound, lemma 5 in \cite{ravikumar2011high}]\label{lemma:remainderbound}
    If \[\Vert{\it \Delta}\Vert_{\max} \leq \frac{1}{3\kappa_{\Sigma}s},\] then we have that \[\Vert R({\it \Delta})\Vert_{\max}\leq \frac{3}{2} s\kappa_{\Sigma}^3\Vert{\it \Delta}\Vert_{\max}^2. \]
\end{lemma}
\begin{lemma}[Support error control]
\label{lem:Delta_control_brouwer}
 Define
\[
\hat\lambda_{\max,S}:=\max_{(i,j)\in S}\hat\lambda_{ij},
\qquad
r:=2\kappa_\Gamma\Big(\|W\|_{\max}+\hat\lambda_{\max,S}\Big).
\]

Assume that
\begin{equation}
\label{eq:r_small_brouwer}
r \le \min\Big\{\frac{1}{3\kappa_\Sigma s},\ \frac{1}{3\kappa_\Sigma^3\kappa_\Gamma s}\Big\},
\end{equation}
then the oracle error satisfies $\|\Delta\|_{\max}\le r$.
\end{lemma}

\begin{proof}
The restricted problem \eqref{eq:restricted_thm} is a convex program over a convex set.
Moreover, the smooth part $C\mapsto -\log\det C + \langle S,C\rangle$ is strictly convex on $C\succ0$.
Therefore, the problem has a unique solution.

We have ${\it\Delta}_{S^c}=0$ by PDW construction.
Restricting the KKT condition of the restricted problem to $S$ and using
$\bar{A}=\Sigma+W$ together with the remainder definition
yields the following
\begin{equation}
\label{eq:G_restricted}
\Big[ W+\Sigma{\it \Delta}\Sigma - R({\it\Delta}) + (\hat\Lambda\odot \dot Z) \Big]_{S}=0.
\end{equation}

We vectorize the left-hand side in \eqref{eq:G_restricted} and use
$\mathrm{vec}(\Sigma{\it \Delta}\Sigma)=\Gamma\mathrm{vec}({\it \Delta})$, and obtain the following 
\[
\Gamma_{SS}\mathrm{vec}({\it \Delta}_{S})
-
\mathrm{vec}\Big(R({\it \Delta})_{S}-W_{S}-(\hat\Lambda\odot \dot Z)_{S}\Big)
=0.
\]

Multiplying the above by $(\Gamma_{SS})^{-1}$ gives
\[
T({\it\Delta})
:=
\mathrm{vec}({\it \Delta}_{S})-
(\Gamma_{SS})^{-1}
\mathrm{vec}\Big(R({\it\Delta})_{S}-W_{S}-(\hat\Lambda\odot \dot Z)_{S}\Big),
\]
so that \eqref{eq:G_restricted} holds if and only if $T({\it\Delta})=0$.

For any $u\in \mathbb{R}^{|S|}$,
let ${\it\Delta}(u)$ denote the symmetric matrix with ${\it\Delta}(u)_{S^c}=0$ and
$\mathrm{vec}({\it\Delta}(u)_S)=u$. Define $F:\mathbb{R}^{|S|}\to\mathbb{R}^{|S|}$ by
\[
F(u)
:=u-T({\it\Delta}(u))
=(\Gamma_{SS})^{-1} 
\mathrm{vec}\Big(R({\it\Delta}(u))_{S}-W_{S}-(\hat\Lambda\odot \dot Z)_{S}\Big).
\]
Then $F(u)=u$ if and only if $T({\it\Delta}(u))=0$, equivalently if and only if
\eqref{eq:G_restricted} holds for ${\it\Delta}(u)$.

Let
\[
\mathcal B(r):=\Big\{u\in \mathbb{R}^{|S|}:\ \|u\|_\infty\le r\Big\}.
\]

Take any $u\in \mathcal B(r)$ and , with a slight abuse of notation, write ${\it\Delta}={\it\Delta}(u)$.

Because $\|\dot Z_{S}\|_{\max}\le 1$ and  the penalty is bounded by
$\hat\lambda_{\max,S}$ on $S$, we have
\[
\|(\hat\Lambda\odot \dot Z)_{S}\|_{\max}\le \hat\lambda_{\max,S}.
\]

By definition of $\kappa_\Gamma=\|(\Gamma_{SS})^{-1}\|_\infty$, we obtain
\[
\|F(u)\|_{\infty}
\le
\kappa_\Gamma\Big(\|R({\it \Delta})\|_{\max}+\|W\|_{\max}+\hat\lambda_{\max,S}\Big).
\]

 Since $\|{\it\Delta}\|_{\max}\le r$, $r\le (3\kappa_\Sigma s)^{-1}$, and ${\it\Delta}_{S^c}=0$, Lemma \ref{lemma:remainderbound} indicates the remainder bound
\[
\|R({\it \Delta})\|_{\max}\le \frac{3}{2} s \kappa_\Sigma^3 \|{\it \Delta}\|_{\max}^2
\le \frac{3}{2} s \kappa_\Sigma^3 r^2.
\]

Therefore,
\[
\|F(u)\|_{\infty}
\le
\kappa_\Gamma\Big(\tfrac{3}{2}s\kappa_\Sigma^3 r^2 + \|W\|_{\max}+\hat\lambda_{\max,S}\Big).
\]

Using $r:=2\kappa_\Gamma(\|W\|_{\max}+\hat\lambda_{\max,S})$, we obtain
\[
\|F(u)\|_{\infty}
\le
\frac{3}{2}\kappa_\Gamma s\kappa_\Sigma^3 r^2 + \frac{r}{2}.
\]

Finally, the condition $r\le (3\kappa_\Sigma^3\kappa_\Gamma s)^{-1}$ implies
\[
\frac{3}{2}\kappa_\Gamma s\kappa_\Sigma^3 r^2 \le \frac{r}{2},
\]
hence $\|F(u)\|_{\infty}\le r$. Therefore, $F$ maps $\mathcal B(r)$ into itself.
Since $F$ is continuous on $\mathcal B(r)$
and $\mathcal B(r)$ is convex and compact, the Brouwer's fixed point theorem
\citep{ortega2000iterative} implies there exists $u^\sharp\in \mathcal B(r)$ such that
$F(u^\sharp)=u^\sharp$.

Let ${\it\Delta}^\sharp={\it\Delta}(u^\sharp)$, i.e. ${\it\Delta}^\sharp_{S^c}=0$ and
$\mathrm{vec}({\it\Delta}^\sharp_S)=u^\sharp$. Then $F(u^\sharp)=u^\sharp$ implies
$T({\it\Delta}^\sharp)=0$, so \eqref{eq:G_restricted} holds for ${\it\Delta}^\sharp$.
By the uniqueness of the restricted optimum, this fixed point must equal the solution
${\it \Delta}=\dot C-\Sigma^{-1}$. Therefore,
\[
\|{\it \Delta}\|_{\max}=\|u^\sharp\|_\infty \le r.
\]

\end{proof}

\begin{lemma}[Strict dual feasibility]
\label{lem:strict_dual_feasibility_weak}
Suppose Assumption~\ref{boundedincoherence} holds.  
Let $\dot Z$ be the PDW dual variable constructed in~\eqref{eq:restricted_thm}--\eqref{eq:Zoff_def_thm}.  
Define
\[
\hat\lambda_{\min,S^c}:=\min_{e\in S^c}\hat\lambda_e ,\quad \hat\psi
:=
\max_{e\in S^c}\sum_{s\in S}|A_{es}|
\frac{\hat\lambda_s}{\hat\lambda_e}.
\]

If $\hat\psi\le 1-5\alpha/8$ and
\begin{equation}
\label{eq:lemma3_sufficient}
\frac{1+\kappa_A}{\hat\lambda_{\min,S^c}}
\Big(\|W\|_{\max}+\|R({\it \Delta})\|_{\max}\Big)
\le \frac{\alpha}{2},
\end{equation}
then $\|\dot Z_{S^c}\|_{\max}<1$ and $\dot C=\widetilde{C}$.
\end{lemma}

\begin{proof}
The vectorized KKT condition for the restricted problem \eqref{eq:restricted_thm} is
\begin{equation}
\label{eq:kkt_vec_master}
\Gamma {\it \Delta} + W - R({\it \Delta}) + \hat\Lambda \dot Z = 0 .
\end{equation}

We partition \eqref{eq:kkt_vec_master} into $(S,S^c)$ blocks. Since our PDW construction enforces
${\it \Delta}_{S^c}=0$, we obtain
\begin{align}
\Gamma_{SS}{\it \Delta}_S + W_S - R({\it \Delta})_S + \hat\Lambda_S \dot Z_S &= 0,
\label{eq:kkt_S}\\
\Gamma_{S^cS}{\it \Delta}_S + W_{S^c} - R({\it \Delta})_{S^c} + \hat\Lambda_{S^c}\dot Z_{S^c} &= 0 .
\label{eq:kkt_Sc}
\end{align}

Solving \eqref{eq:kkt_S} gives
\begin{equation}
\label{eq:Delta_S_solution}
{\it \Delta}_S
=
-(\Gamma_{SS})^{-1}
\Big(W_S - R({\it \Delta})_S + \hat\Lambda_S \dot Z_S \Big).
\end{equation}

Recall
\[
A := \Gamma_{S^cS}(\Gamma_{SS})^{-1}.
\]

Substituting \eqref{eq:Delta_S_solution} into \eqref{eq:kkt_Sc} yields
\begin{equation}
\label{eq:ZSc_exact}
\hat\Lambda_{S^c}\dot Z_{S^c}
=
A\Big(W_S - R({\it \Delta})_S + \hat\Lambda_S \dot Z_S\Big)
-\Big(W_{S^c}-R({\it \Delta})_{S^c}\Big).
\end{equation}

Fix any $e\in S^c$. Taking absolute values of the $e$th component of \eqref{eq:ZSc_exact} gives
\begin{equation}
\label{eq:ZSc_component_bound}
|\dot Z_e|
\le
\frac{1}{\hat\lambda_e}
\Big(
|[A(W_S-R({\it \Delta})_S)]_e|
+|W_e-R({\it \Delta})_e|
+|[A\hat\Lambda_S\dot Z_S]_e|
\Big).
\end{equation}

Write row $e$ of $A$ as $(A_{es})_{s\in S}$. Then
\[
|[A(W_S-R({\it \Delta})_S)]_e|
\le
\sum_{s\in S}|A_{es}|(|W_s|+|R({\it \Delta})_s|)
\le
\|A_{e,\cdot}\|_1
(\|W\|_{\max}+\|R({\it \Delta})\|_{\max}).
\]

Also,
\[
|W_e-R({\it \Delta})_e|
\le \|W\|_{\max}+\|R({\it \Delta})\|_{\max}.
\]

Recall
\[
\kappa_A = \max_{e\in S^c}\|A_{e,\cdot}\|_1 .
\]

Hence, uniformly over $e\in S^c$,
\begin{equation}
\label{eq:noise_uniform}
\frac{1}{\hat\lambda_e}
\Big(
|[A(W_S-R({\it \Delta})_S)]_e|
+|W_e-R({\it \Delta})_e|
\Big)
\le
\frac{1+\kappa_A}{\hat\lambda_{\min,S^c}}
(\|W\|_{\max}+\|R({\it \Delta})\|_{\max}).
\end{equation}

Since $|\dot Z_s|\le 1$ for $s\in S$,
\[
|[A\hat\Lambda_S\dot Z_S]_e|
=
\left|\sum_{s\in S}A_{es}\hat\lambda_s \dot Z_s\right|
\le
\sum_{s\in S}|A_{es}|\hat\lambda_s .
\]

Then
\begin{equation}
\label{eq:leakage_uniform}
\frac{1}{\hat\lambda_e}
|[A\hat\Lambda_S\dot Z_S]_e|
\le \hat\psi.
\end{equation}

Combining \eqref{eq:noise_uniform} and \eqref{eq:leakage_uniform} in
\eqref{eq:ZSc_component_bound} and taking $\max_{e\in S^c}$ yields
\begin{equation}
\label{eq:lemma3_bound}
\|\dot Z_{S^c}\|_{\max}
\le
\frac{1+\kappa_A}{\hat\lambda_{\min,S^c}}
\Big(\|W\|_{\max}+\|R({\it \Delta})\|_{\max}\Big)
+\hat\psi .
\end{equation}

If \eqref{eq:lemma3_sufficient} holds and $\hat\psi\le 1-5\alpha/8$ hold, then
\[
\|\dot Z_{S^c}\|_{\max}
\le \frac{\alpha}{2} + \left(1-\frac{5\alpha}{8} \right) < 1.
\]

\end{proof}
\begin{lemma}[Sign recovery]
\label{lem:sign}
If $\|\widetilde C-\Sigma^{-1}\|_{\max}\le \min_{(i,j)\in E}|[\Sigma^{-1}]_{ij}|/2$,
then $\mathrm{sign}(\widetilde C_{ij})=\mathrm{sign}([\Sigma^{-1}]_{ij})$ for all $(i,j)\in E$.
\end{lemma}
\begin{lemma}
\label{lem:psi_transfer_explicit}
 If we have $\max_{1\leq i\leq d}\left|\hat\omega_i-\omega^\ast_i\right|\le t_\omega$ and
\[
\varepsilon_n:=\frac{2t_\omega}{\delta}<1,
\]
 then for all off-diagonal indices $s,e$,
\[
\frac{\hat\lambda_s}{\hat\lambda_e}
\le \frac{1+\varepsilon_n}{1-\varepsilon_n}\frac{\lambda^\ast_s}{\lambda^\ast_e},
\qquad\text{and hence}\qquad
\hat\psi\le \frac{1+\varepsilon_n}{1-\varepsilon_n}\psi^\ast.
\]

In particular, if $\psi^\ast\le 1-\alpha$ and $\varepsilon_n\le \alpha/8$, then
$\hat\psi\le 1-5\alpha/8$.
\end{lemma}

\begin{proof}
For any $i\ne j$,
\[
|\hat\lambda_{ij}-\lambda^\ast_{ij}|=\delta\big|(\hat\omega_i-\omega^\ast_i)+(\hat\omega_j-\omega^\ast_j)\big|
\le 2\delta t_\omega.
\]

Also $\lambda^\ast_{ij}\ge \delta^2$, we have
\[
\frac{|\hat\lambda_{ij}-\lambda^\ast_{ij}|}{\lambda^\ast_{ij}}
\le \frac{2\delta t_\omega}{\delta^2}=\frac{2t_\omega}{\delta}=\varepsilon_n.
\]
Thus, \[(1-\varepsilon_n)\lambda^\ast_{ij}\le \hat\lambda_{ij}\le (1+\varepsilon_n)\lambda^\ast_{ij},\]
implying \[\frac{\hat\lambda_s}{\hat\lambda_e} \le \frac{1+\varepsilon_n}{1-\varepsilon_n}\frac{\lambda_s^\ast}{\lambda_e^\ast}.\]
Taking maxima in the definition of $\hat\psi,\psi^\ast$ yields the inequality 
\[\hat\psi\le \frac{1+\varepsilon_n}{1-\varepsilon_n}\psi^\ast.
\]

If $\varepsilon_n\le \alpha/8$, then \[\frac{1+\varepsilon_n}{1-\varepsilon_n}\le 1+3\varepsilon_n\le 1+\frac{3\alpha}{8},\]
and therefore \[\hat\psi \le \left(1+\frac{3\alpha}{8}\right)(1-\alpha)\le 1-\frac{5\alpha}{8}.\]
\end{proof}

\begin{lemma}[Concentration Inequality I]\label{concentration1}
There exist  universal constants $C_1>0, c_1>0$ such that for any $\tau>2$,
\[
\mathbb{P}\left(\|W\|_{\max} \le C_1\sigma_{\max}^{2}\sqrt{\frac{\tau\log d}{n}}\right)
\ \ge\ 1-2d^{2-\tau}
\] holds when $n\geq  c_1\tau \log d.$
\end{lemma}

\begin{proof}
Let $Y_t:=\bm{x}^{t}_i \bm{x}^{t}_j-\Sigma_{ij}$. Under Gaussianity,
$Y_t$ is sub-exponential with $\|Y_t\|_{\psi_1}\le K_1\sigma_{\max}^{2}$ for a universal $K_1$.
The Bernstein concentration inequality for sub-exponential variables gives
\begin{equation}\label{concen1}
\mathbb{P}\left(\left|\frac1n\sum_{t=1}^n Y_t\right|\ge u\right)
\le 2\exp\!\left(-C_0 n \min\Big\{\frac{u^2}{K_1^2\sigma_{\max}^{4}},\frac{u}{K_1\sigma_{\max}^{2}}\Big\}\right).
\end{equation}

We take
\[
u=C_1\sigma_{\max}^{2}\sqrt{\frac{\tau\log d}{n}}.
\]
If 
\[
n\ge \Big(\frac{C_1}{K_1}\Big)^2 \tau\log d,
\]
we have  $u\le K_1\sigma_{\max}^2$, so the minimum in the concentration inequality \eqref{concen1} is achieved by the quadratic term. Hence
\[
\mathbb{P}(|W_{ij}|\ge u)
\le 2\exp\!\left(-C_0 n \frac{u^2}{K_1^2\sigma_{\max}^{4}}\right)
=2\exp\!\left(-C_0\frac{C_1^2}{K_1^2}\tau\log d\right).
\]

Choosing $C_1\ge K_1/\sqrt{C_0}$ results in $\mathbb{P}(|W_{ij}|\ge u)\le 2d^{-\tau}$.
The union bound over $d^2$ pairs gives the result.

\end{proof}

\begin{lemma}[Concentration Inequality II]
\label{concentration2}
There exist a universal constant $C_2> 0$ such that for any $\tau > 1$,
\[
\mathbb{P}\!\left(\max_{1\le i\le d}\left|\mathbb{E}_{\mathbb{P}_n} [\vert\bm{x}_i\vert] - \mathbb{E}[\vert\bm{x}_i\vert]\right|
\le C_2\sigma_{\max}\sqrt{\frac{\tau\log d}{n}}\right)
\ \ge\ 1-2d^{1-\tau}.
\]
\end{lemma}

\begin{proof}
Let $Z_t := \vert\bm{x}^t_i\vert - \mathbb{E}[\vert\bm{x}_i\vert]$.  $Z_t$ is a sub-Gaussian random variable with
$\|Z_t\|_{\psi_2} \le K_2 \sigma_{\max}$
for some universal constant $K_2$.

We use the Hoeffding-type concentration bound:
\[
\mathbb{P}\left(\left|\frac{1}{n}\sum_{t=1}^n Z_t\right| \ge u\right)
\le 2\exp\!\left(-\frac{c n u^2}{K_2^2 \sigma_{\max}^2}\right).
\]
Take \[u = C_2 \sigma_{\max} \sqrt{\frac{\tau \log d}{n}},\quad C_2 = \frac{K_2}{\sqrt{c}},\]  and substitute into the exponent:
\[
-\frac{c n}{K_2^2 \sigma_{\max}^2} \left( C_2^2 \sigma_{\max}^2 \frac{\tau \log d}{n} \right) = - \tau \log d.
\]

Applying the union bound over $d$ indices gives the result.
\end{proof}

Equipped with the lemmas above, we proceed to prove Theorem \ref{thm:main}.
\begin{proof}
Define the event
\[
\mathcal E
:=
\left\{\|W\|_{\max} \le C_1\sigma_{\max}^{2}\sqrt{\frac{\tau\log d}{n}}\right\}
\cap
\left\{\max_i|\hat\omega_i-\omega_i^\ast|\le C_2\sigma_{\max}\sqrt{\frac{\tau\log d}{n}}\right\}.
\]

By Lemma \ref{concentration1}, Lemma \ref{concentration2} and a union bound,
\[
\mathbb{P}(\mathcal E)\ge 1-2d^{2-\tau}-2d^{1-\tau}\ge 1-4d^{2-\tau}.
\]

We complete our proof in six steps stated below.

\noindent\textbf{Step 1: control $\hat\psi$.}

On the event $\mathcal E$, we define
\[
t_\omega:=C_2\sigma_{\max}\sqrt{\frac{\tau\log d}{n}},
\qquad
\varepsilon_n:=\frac{2t_\omega}{\delta}=\frac{2C_2\sigma_{\max}\sqrt{\tau}}{c_\delta}.
\]

By the definition of $c_\delta$, we have $$\varepsilon_n\le \frac{\alpha}{8},$$ then Lemma \ref{lem:psi_transfer_explicit} implies
\[
\hat\psi \le 1-\frac{5\alpha}{8}.
\]

\noindent\textbf{Step 2: lower bound $\hat\lambda_{\min,S^c}$.}

For any off-diagonal $e=(i,j)\in S^c$,
\[
\hat\lambda_e
=\delta(\hat\omega_i+\hat\omega_j)+\delta^2
\ge \delta\big((\omega_i^\ast-t_\omega)+(\omega_j^\ast-t_\omega)\big)+\delta^2
\ge 2\delta(\omega_{\min}^\ast-t_\omega)+\delta^2.
\]

Since we have \[n\ge 2\pi C_2^2\left(\frac{\sigma_{\max}}{\sigma_{\min}}\right)^2 \tau\log d,\] then
\[
t_\omega=C_2\sigma_{\max}\sqrt{\frac{\tau\log d}{n}}
\le \frac{1}{2}\sqrt{\frac{2}{\pi}} \sigma_{\min}=\frac{\omega_{\min}^\ast}{2},
\]
and therefore
\begin{equation}
\label{eq:lambda_min_lower_thm}
\hat\lambda_{\min,S^c}:=\min_{e\in S^c}\hat\lambda_e \ge\ \omega_{\min}^\ast \delta.
\end{equation}

\noindent\textbf{Step 3: upper bound $\hat\lambda_{\max,S}$.}
\[
\hat\lambda_{ij}
=\delta(\hat\omega_i+\hat\omega_j)+\delta^2
\le 2\delta(\omega_{\max}^\ast+t_\omega)+\delta^2.
\]

We have $t_\omega\le \omega_{\min}^\ast/2\le \omega_{\max}^\ast/2$, hence
\[
\hat\lambda_{\max,S}:=\max_{(i,j)\in S}\hat\lambda_{ij}
\le 3\omega_{\max}^\ast \delta+\delta^2.
\]
Combining with $\|W\|_{\max}\le C_1\sigma_{\max}^2\sqrt{\tau\log d/n}$, $\delta=c_\delta\sqrt{\log d/n}$, and  $n\ge \log d$, we obtain
\[
\|W\|_{\max}+\hat\lambda_{\max,S}
\le \Big(C_1\sigma_{\max}^{2}\sqrt{\tau}+3\omega_{\max}^\ast c_\delta+c_\delta^2\Big)\sqrt{\frac{\log d}{n}}
= B\sqrt{\frac{\log d}{n}}.
\]
Define as in Lemma~\ref{lem:Delta_control_brouwer}
\[
r:=2\kappa_\Gamma\Big(\|W\|_{\max}+\hat\lambda_{\max,S}\Big),
\]
    then we can have that
\begin{equation}
\label{eq:r_bound_thm}
r\ \le\ 2\kappa_\Gamma B \sqrt{\frac{\log d}{n}}.
\end{equation}

\noindent\textbf{Step 4: support error and remainder control.}

The sample size condition  and \eqref{eq:r_bound_thm} implies 
\[
r\le \frac{1}{3\kappa_\Sigma s},
\quad
r\le \frac{1}{3\kappa_\Sigma^3\kappa_\Gamma s},
\]
so the condition \eqref{eq:r_small_brouwer} in Lemma \ref{lem:Delta_control_brouwer} holds.

Hence, the oracle error follows
\begin{equation}
\label{eq:Delta_bound_thm}
\|{\it\Delta}\|_{\max}\le r\le 2\kappa_\Gamma B \sqrt{\frac{\log d}{n}}.
\end{equation}

Then, Lemma \ref{lemma:remainderbound} gives
\begin{equation}\label{remainderboundequ}
\|R({\it\Delta})\|_{\max}
\le \frac{3}{2} s \kappa_\Sigma^3 \|\Delta\|_{\max}^2
\le \frac{3}{2} s \kappa_\Sigma^3 r^2
\le  6 s \kappa_\Sigma^3 \kappa_\Gamma^2 B^2 \frac{\log d}{n}.
\end{equation}

\noindent\textbf{Step 5: strict dual feasibility.}

By Lemma~\ref{lem:strict_dual_feasibility_weak}, it suffices to check:
\[
\hat\psi\le 1-\frac{5\alpha}{8},
\quad
\frac{1+\kappa_A}{\hat\lambda_{\min,S^c}}\Big(\|W\|_{\max}+\|R({\it\Delta})\|_{\max}\Big)\le \frac{\alpha}{2}.
\]
The first inequality is proved in Step 1.

For the second inequality, we split the left-hand side into the $W$-term and the $R$-term.

(a) It follows from \eqref{eq:lambda_min_lower_thm} and $\|W\|_{\max}\le C_1\sigma_{\max}^2\sqrt{\tau\log d/n}$ that
\[
\frac{1+\kappa_A}{\hat\lambda_{\min,S^c}}\|W\|_{\max}
\le 
\frac{1+\kappa_A}{\omega_{\min}^\ast\delta} C_1\sigma_{\max}^2\sqrt{\frac{\tau\log d}{n}}
=
\frac{(1+\kappa_A)C_1\sigma_{\max}^2\sqrt{\tau}}{\omega_{\min}^\ast c_\delta}.
\]

Since $\omega_{\min}^\ast=\sqrt{2/\pi} \sigma_{\min}$, and by the definition of $c_\delta$,
we get
\[
\frac{(1+\kappa_A)C_1\sigma_{\max}^2\sqrt{\tau}}{\omega_{\min}^\ast c_\delta}
\le
\frac{\alpha}{4}.
\]

(b) It follows from \eqref{eq:lambda_min_lower_thm} and \eqref{remainderboundequ} that
\[
\frac{1+\kappa_A}{\hat\lambda_{\min,S^c}}\|R({\it\Delta})\|_{\max}
\le
\frac{6(1+\kappa_A) s \kappa_\Sigma^3 \kappa_\Gamma^2 B^2}{\omega_{\min}^\ast c_\delta}\sqrt{\frac{\log d}{n}}.
\]

The condition of the sample size $n$ implies
\[
\sqrt{\frac{\log d}{n}}
\le
\frac{\alpha \omega_{\min}^\ast c_\delta}{24(1+\kappa_A) s \kappa_\Sigma^3 \kappa_\Gamma^2 B^2},
\]
and 
\[
\frac{6(1+\kappa_A) s \kappa_\Sigma^3 \kappa_\Gamma^2 B^2}{\omega_{\min}^\ast c_\delta}\sqrt{\frac{\log d}{n}}
\le \frac{\alpha}{4}.
\]

Combining (a) and (b) indicates
\[
\frac{1+\kappa_A}{\hat\lambda_{\min,S^c}}\Big(\|W\|_{\max}+\|R({\it\Delta})\|_{\max}\Big)\le \frac{\alpha}{2}.
\]

Therefore, we apply Lemma~\ref{lem:strict_dual_feasibility_weak} and obtain strict dual feasibility
$\|\dot Z_{S^c}\|_{\max}<1$, implying the PDW solution coincides with the global optimum:
\[
\widetilde C=\dot C \quad\text{and}\quad \widetilde  C_{S^c}=0.
\]

\noindent\textbf{Step 6: convergence rate and sign recovery.}

Since $\widetilde  C=\dot C$ and $\dot C-\Sigma^{-1}={\it\Delta}$,  it follows from \eqref{eq:Delta_bound_thm} that,
\[
\|\widetilde C-\Sigma^{-1}\|_{\max}=\|{\it\Delta}\|_{\max}\le 2\kappa_\Gamma B \sqrt{\frac{\log d}{n}}.
\]

Finally, if \[\min_{(i,j)\in E}|[\Sigma^{-1}]_{ij}|\ge 4\kappa_\Gamma B\sqrt{\frac{\log d}{n}},\] then 
$\|\widetilde C-\Sigma^{-1}\|_{\max}\le \min_{(i,j)\in E}|[\Sigma^{-1}]_{ij}|/2$, and Lemma~\ref{lem:sign} yields sign recovery.
\end{proof}
\section{Proof of Proposition \ref{prop:scale-inflation-mitigation}}
\begin{proof}
By the property of the Kronecker product,
\[
\Gamma=\Sigma\otimes\Sigma=(DRD)\otimes(DRD)=(D\otimes D)(R\otimes R)(D\otimes D)
=(D\otimes D) \Gamma^{\dagger} (D\otimes D).
\]

Define $D_\Gamma:=D\otimes D$.  
Then, we have that
\[
\Gamma_{S^cS}=(D_\Gamma)_{S^cS^c} \Gamma^{\dagger}_{S^cS} (D_\Gamma)_{SS},
\quad
\Gamma_{SS}=(D_\Gamma)_{SS} \Gamma^{\dagger}_{SS} (D_\Gamma)_{SS}.
\]

Since $(D_\Gamma)_{SS}$ is diagonal and invertible, we have that
\[
(\Gamma_{SS})^{-1}
=\big((D_\Gamma)_{SS} \Gamma^{\dagger}_{SS} (D_\Gamma)_{SS}\big)^{-1}
=(D_\Gamma)_{SS}^{-1} (\Gamma^{\dagger}_{SS})^{-1} (D_\Gamma)_{SS}^{-1},
\]
\begin{equation}\label{eq:Astar-factor}
A
=\Gamma_{S^cS}(\Gamma_{SS})^{-1}=(D_\Gamma)_{S^cS^c} 
\Gamma^{\dagger}_{S^cS}(\Gamma^{\dagger}_{SS})^{-1}
(D_\Gamma)_{SS}^{-1}=(D_\Gamma)_{S^cS^c} A^{\dagger} (D_\Gamma)_{SS}^{-1}.
\end{equation}

Let $d_u:=\big[D_\Gamma\big]_{uu}$.
It follows from \eqref{eq:Astar-factor} that
\[
A_{es}=d_e A^{\dagger}_{es} d_s^{-1}.
\]

Hence,
\[|A_{es}|=|A^{\dagger}_{es}|\frac{d_e}{d_s}.\]

Therefore, 
\[
\max_{e\in S^c}\sum_{s\in S}|A_{es}|
\le \max_{e\in S^c} d_e\max_{s\in S} d_s^{-1}\|A^{\dagger}\|_\infty.
\]

Since we have $d_u=\sqrt{\Sigma_{ii}}\sqrt{\Sigma_{jj}}$, we have that
\[
\max_{e\in S^c}\sum_{s\in S}|A_{es}|
\le \Big(\frac{\sigma_{\max}}{\sigma_{\min}}\Big)^2\|A^{\dagger}\|_\infty,
\]
which proves the first inequality.

By definition, we have that
\[
\psi^\ast=\max_{e\in S^c}\sum_{s\in S}|A_{es}|\frac{\lambda_s^\ast}{\lambda_e^\ast}
=\max_{e\in S^c}\sum_{s\in S}|A^{\dagger}_{es}|\frac{d_e}{d_s}\frac{\lambda_s^\ast}{\lambda_e^\ast}.
\]

If we define $w_u:=\lambda_u^\ast/d_u$, then we have that
\begin{equation}\label{defofpsi}
\psi^\ast=\max_{e\in S^c}\sum_{s\in S}|A^{\dagger}_{es}|\frac{w_s}{w_e}
\le \frac{\max_{s\in S}w_s}{\min_{e\in S^c}w_e} \max_{e\in S^c}\sum_{s\in S}|A^{\dagger}_{es}|
=\frac{\max_{s\in S}w_s}{\min_{e\in S^c}w_e} \|A^{\dagger}\|_\infty.
\end{equation}

Recall that for $i\neq j$,
\[
\lambda_{ij}^\ast=\delta(\omega_i^\ast+\omega_j^\ast)+\delta^2,
\qquad \omega_i^\ast=\sqrt{2/\pi} \sqrt{\Sigma_{ii}}.
\]

 Then for $i\neq j$,
\[
w_{(i,j)}=\frac{\lambda_{ij}^\ast}{\sqrt{\Sigma_{ii}}\sqrt{\Sigma_{jj}}}
=\delta \sqrt{\frac{2}{\pi}}\Big(\frac{1}{\sqrt{\Sigma_{ii}}}+\frac{1}{\sqrt{\Sigma_{jj}}}\Big)+\frac{\delta^2}{\sqrt{\Sigma_{ii}}\sqrt{\Sigma_{jj}}}.
\]

Using $\sigma_{\min}\le \sqrt{\Sigma_{ii}},\sqrt{\Sigma_{jj}}\le \sigma_{\max}$, we have the bounds
\[
w_{(i,j)}\le \delta \sqrt{\frac{2}{\pi}}\frac{2}{\sigma_{\min}}+\frac{\delta^2}{\sigma_{\min}^2},
\qquad
w_{(i,j)}\ge \delta \sqrt{\frac{2}{\pi}} \frac{2}{\sigma_{\max}}+\frac{\delta^2}{\sigma_{\max}^2}.
\]

Therefore,
\begin{align*}
\frac{\max_{s\in S}w_s}{\min_{e\in S^c}w_e}
\le
\Big(\frac{\sigma_{\max}}{\sigma_{\min}}\Big)^2
\frac{2\delta \sqrt{\frac{2}{\pi}} \sigma_{\min}+\delta^2}{2\delta \sqrt{\frac{2}{\pi}} \sigma_{\max}+\delta^2}.
\end{align*}

Since $10\delta<\sqrt{2/\pi}\sigma_{\min}$,
\[
2\delta \sqrt{\frac{2}{\pi}} \sigma_{\min}+\delta^2 \le 21\delta \sqrt{\frac{2}{\pi}} \sigma_{\min}/10,
\qquad
2\delta \sqrt{\frac{2}{\pi}} \sigma_{\max}+\delta^2 \ge 2\delta \sqrt{\frac{2}{\pi}} \sigma_{\max},
\]
and hence
\[
\frac{2\delta \sqrt{\frac{2}{\pi}} \sigma_{\min}+\delta^2}{2\delta \sqrt{\frac{2}{\pi}} \sigma_{\max}+\delta^2}
\le\frac{21}{20}\frac{\sigma_{\min}}{\sigma_{\max}}.
\]

Then, we have that 
\[
\frac{\max_{s\in S}w_s}{\min_{e\in S^c}w_e}
\le \frac{21}{20}\frac{\sigma_{\max}}{\sigma_{\min}}.
\]

Plugging into \eqref{defofpsi} gives
\[
\psi^\ast \le \frac{21}{20}\Big(\frac{\sigma_{\max}}{\sigma_{\min}}\Big)\|A^{\dagger}\|_\infty,
\]
which proves the second inequality.
\end{proof}

\section{Proof of Proposition \ref{regularizationprop}}
We first state two propositions before we prove Proposition \ref{regularizationprop}.
\begin{proposition}\label{peturbproposition}
 For vector $b\in\mathbb{R}^d$,  positive definite matrix $Q\in\mathbb{R}^{d\times d}$ and $p\in(1,\infty)$,  we have that
    \begin{equation}\label{perturbedprob}\max_{\Vert \Delta\Vert_p \leq 1} b^{\top}\Delta +\varepsilon \Delta^\top Q \Delta= \Vert b \Vert_q+\varepsilon \widetilde{\Delta}^{\top} Q  \widetilde{\Delta}+o(\varepsilon),\end{equation}
    as $\varepsilon\to 0$, where $1/p+1/q=1$, and $\widetilde{\Delta}$ is the unique solution to the problem \[\max_{\Vert \Delta \Vert_p \leq 1} b^{\top} \Delta.\]
\end{proposition}
\begin{proof}
We first prove that the optimization problem \eqref{perturbedprob} has a unique solution when $\varepsilon\to0$.

The Lagrangian function of problem  \eqref{perturbedprob} is 
\[ L_{\varepsilon}(\Delta, \lambda)= b^{\top}\Delta +\varepsilon \Delta^\top Q \Delta -\lambda (\Vert \Delta\Vert_p^p-1).\]

Since the constraint set $\Vert \Delta\Vert_p \leq 1$ is convex, it follows from the KKT condition that the solution to the problem, denoted by $\Delta_\ast$ and $\lambda_\ast$, follows the equations:
\begin{equation} \label{KKT}\Vert\Delta_\ast\Vert_p^p-1=0,
\quad b +\varepsilon 2Q\Delta_\ast - \lambda_\ast p\Vert  \Delta_\ast\Vert_p^{p-1}\vert\Delta_\ast\vert^{p-1}\odot \mathrm{sign} ( \Delta_\ast)=0.\end{equation}

Assume there are two distinct solutions to problem \eqref{perturbedprob} when $\varepsilon\to0$, denoted by $\Delta_\ast^1$ and $\Delta_\ast^2$.
Equivalently, both  $\Delta_\ast^1$ and $\Delta_\ast^2$ satisfy the equations in \eqref{KKT}, indicating 
\begin{equation}\label{contradict}2\varepsilon Q(\Delta_\ast^1- \Delta_\ast^2)=\lambda_\ast p (\vert \Delta_\ast^2\vert^{p-1}\odot \mathrm{sign}(\Delta_\ast^2)-\vert \Delta_\ast^1\vert^{p-1}\odot \mathrm{sign}(\Delta_\ast^1)),\end{equation}
\[\Vert\Delta_\ast^1\Vert_p=\Vert\Delta_\ast^2\Vert_p=1,\quad \Delta_\ast^1\not =\Delta_\ast^2\]

It is  easy to see that the left-hand side of \eqref{contradict} goes to $\bm{0}$ when $\varepsilon\to0$ while the right-hand side of \eqref{contradict} strictly not equal to $\bm{0}$ when $p\in(1,\infty)$ and $\Delta_\ast^1\not =\Delta_\ast^2$. This produces a contradiction. Thus, the optimization problem \eqref{perturbedprob} has a unique solution when $\varepsilon\to0$ and $p\in(1,\infty)$.

Then, since the objective function is convex, continuous, and differentiable w.r.t $\varepsilon$ for fixed $\Delta$, by Danskin's theorem \citep{bonnans2013perturbation}, if we denote \[\phi_{b,Q}(\varepsilon)=\max_{\Vert \Delta\Vert_p \leq 1} b^{\top}\Delta +\varepsilon \Delta^\top Q \Delta,\quad  \Delta_{b,Q}(\varepsilon) = \arg\max_{\Vert \Delta\Vert_p \leq 1} b^{\top}\Delta +\varepsilon \Delta^\top Q \Delta,\] then we have that 
\begin{equation*}\begin{aligned}\frac{d}{d\varepsilon}\phi_{b,Q}(\varepsilon)= \Delta_{b,Q}(\varepsilon) ^\top Q \Delta_{{b,Q}}(\varepsilon) ,\end{aligned}\end{equation*}

and 
\begin{equation*}\begin{aligned}\frac{d}{d\varepsilon}\phi_{b,Q}(\varepsilon)\Big\vert_{\varepsilon=0}= \widetilde{\Delta}^\top Q \widetilde{\Delta} ,\end{aligned}\end{equation*}

It follows from Taylor's expansion that 
\[\phi_{b,Q}(\varepsilon)= \phi_{b,Q}(0)+ \varepsilon \widetilde{\Delta}^\top Q \widetilde{\Delta}+o(\varepsilon)= \Vert b\Vert_{q}+ \varepsilon \widetilde{\Delta}^\top Q \widetilde{\Delta}+o(\varepsilon),\]
as $\varepsilon\to 0$.
\end{proof}

\begin{proposition}\label{peturbproposition2}
 For vector $b\in\mathbb{R}^d$ with $[b]_i\not= 0$ for $i=1,\ldots,d$,  positive definite matrix $Q\in\mathbb{R}^{d\times d}$,  we have that
    \begin{equation}\label{perturbedprob2}\max_{\Vert \Delta\Vert_\infty \leq 1} b^{\top}\Delta +\varepsilon \Delta^\top Q \Delta= \Vert b \Vert_1+\varepsilon \mathrm{sign}(b)^{\top} Q  \mathrm{sign}(b),\end{equation}
    as $\varepsilon\to 0$.
\end{proposition}
\begin{proof}
When $\varepsilon =0$,  the optimization problem $\max_{\Vert \Delta\Vert_\infty \leq 1} b^{\top}\Delta$ has the unique solution at $\mathrm{sign}(b)$. 

When $\varepsilon \not =0$, the objective function is convex and the constraint set is a convex compact set. In this way, the solution to the optimization problem is taken at the boundary of the constraint set, i.e., the solution $\Delta^\ast$ satisfies that $[\Delta^\ast]_i = \{-1,1\}$. Since $[b]_i\not =0$, when $\varepsilon$ is very small, flipping the sign of  $\Delta^\ast$ will decrease the value of the objective function. Thus, the optimal solution is  $\mathrm{sign}(b)$ when $\varepsilon\to 0$.
\end{proof}

Now we begin to prove Proposition \ref{regularizationprop}.
\begin{proof}
Notice we have that 
\[\max_{\Vert \bm{\Delta}\Vert_p\leq \delta}(\bm{x}+\bm{\Delta})^\top  C (\bm{x}+\bm{\Delta})= \bm{x}^\top  C \bm{x}+\delta \max_{\Vert \bm{\Delta}\Vert_p\leq 1} \left\{2( C\bm {x})^\top \bm{\Delta}+\delta \bm{\Delta}^\top  C \bm{\Delta}\right\}.\]

% Proposition \ref{peturbproposition} and Proposition \ref{peturbproposition2} imply that it remains to show that 
% \[\mathbb{E}_{\bm{P}}\left[\max_{\Vert \bm{\Delta}\Vert_p\leq 1} \left\{2( C\bm {x})^\top \bm{\Delta}+\delta \bm{\Delta}^\top  C \bm{\Delta}\right\}\right] =2\mathbb{E}_{\bm{P}} \left[ \Vert  C\bm {x}\Vert_{q}\right] +\delta \mathbb{E}_{\bm{P}} \left[ \bm{v}_{\bm{x}}^\top  C \bm{v}_{\bm{x}}\right] +o(\delta) \]

In the probability space $(\bm{P},  C, \mathcal{F})$ for variable $\bm{x}$, for each $\bm{\omega}\in  C$,
it follows from Proposition \ref{peturbproposition} and Proposition \ref{peturbproposition2} that 
\[\max_{\Vert \bm{\Delta}\Vert_p\leq 1} \left\{2( C\bm {x}(\bm{\omega}))^\top \bm{\Delta}+\delta \bm{\Delta}^\top  C \bm{\Delta}\right\} = 2\Vert  C\bm {x}(\bm{\omega})\Vert_{q}+ \delta \bm{v}_{\bm{x}(\bm{\omega})}^\top  C \bm{v}_{\bm{x}(\bm{\omega})} +\mathcal{R}(\bm{\omega}),\]
where \[\lim_{\delta\to 0 } \frac{\mathcal{R}(\bm{\omega})}{\delta}=0. \]

Since we have that 
\begin{equation}\begin{aligned}\vert \mathcal{R}(\bm{\omega})\vert& = \left\vert\max_{\Vert \bm{\Delta}\Vert_p\leq 1} \left\{2( C\bm {x}(\bm{\omega}))^\top \bm{\Delta}+\delta \bm{\Delta}^\top  C \bm{\Delta}\right\}- 2\Vert  C\bm {x}(\bm{\omega})\Vert_{q}+ \delta \bm{v}_{\bm{x}(\bm{\omega})}^\top  C \bm{v}_{\bm{x}(\bm{\omega})}\right\vert\\
&\leq 4 \Vert  C\bm {x}(\bm{\omega})\Vert_{q}+2\delta \max_{\Vert y \Vert_p\leq 1,y\in\mathbb{R}^d} y^\top  C y\\
&\leq 4 \Vert  C\Vert_{p\to q} \Vert\bm{x}(\bm{\omega})\Vert_p+2\delta \Vert  C\Vert_{p\to q},\end{aligned}\end{equation}
implying that $ \mathcal{R}(\bm{\omega})$ is upper bounded by an integral function. Due to the dominated convergence theorem, we have that 
\[\lim_{\delta\to 0}\frac{\mathbb{E}_{\bm{P}}[ \mathcal{R}] }{\delta} = \int_{\bm{\omega}\in  C} \lim_{\delta\to 0}\frac{\mathcal{R}(\bm{\omega})}{\delta} d\bm{P}(\bm{\omega})=0.\]
\end{proof}

\section{Proof of Theorem \ref{asymptoticresults1}}
\begin{proof}        
We define \[\Psi_n( C)= -\log \det ( C)+\mathbb{E}_{\mathbb{P}_n} \left[\max_{\Vert \bm{\Delta}\Vert_p\leq \delta}(\bm{x}+\bm{\Delta})^\top  C (\bm{x}+\bm{\Delta})\right].\]

For $ C\succ 0$, the function $\Psi_n( C)$  is convex since we have the following:
\begin{equation}\begin{aligned}&\Psi_n(\alpha C_1+(1-\alpha) C_2)\\
=& -\log \det (\alpha C_1+(1-\alpha) C_2) + \mathbb{E}_{\mathbb{P}_n} \left[\max_{\Vert \bm{\Delta}\Vert_p\leq \delta}(\bm{x}+\bm{\Delta})^\top (\alpha C_1+(1-\alpha) C_2) (\bm{x}+\bm{\Delta})\right]\\
\leq& -\alpha \log \det (C_1) -(1-\alpha)\log \det (C_2)\\
+&\alpha\mathbb{E}_{\mathbb{P}_n} \left[\max_{\Vert \bm{\Delta}\Vert_p\leq \delta}(\bm{x}+\bm{\Delta})^\top  C_1  (\bm{x}+\bm{\Delta}) \right] +  (1-\alpha)\mathbb{E}_{\mathbb{P}_n} \left[\max_{\Vert \bm{\Delta}\Vert_p\leq \delta}(\bm{x}+\bm{\Delta})^\top  C_2  (\bm{x}+\bm{\Delta}) \right],\end{aligned}\end{equation}
where the second inequality comes from the convexity of the function $-\log\det( C)$.

\noindent\textbf{Case 1}: $\gamma=1/2$.

It follows from Proposition \ref{regularizationprop} that 
\begin{equation*}
    \begin{aligned}
          &V_n(\bm{U})\\ \overset{\vartriangle}{=} &n\left(\Psi_n(\Sigma^{-1}+\frac{1}{\sqrt{n}}\bm{U})-\Psi (\Sigma^{-1})\right)\\     
          =&-n\log \det (\Sigma^{-1}+ \frac{1}{\sqrt{n}}\bm{U})+n\log \det (\Sigma^{-1})+ \sqrt{n}\mathbb{E}_{\mathbb{P}_n} \left[\bm{x}^\top  \bm{U} \bm{x}\right]\\
         &+2\eta\sqrt{n}\left(\mathbb{E}_{\mathbb{P}_n} \left[ \left\Vert (\Sigma^{-1}+\frac{1}{\sqrt{n}}\bm{U})\bm{x}\right\Vert_{q}\right]-\mathbb{E}_{\mathbb{P}_n}\left[ \left\Vert  \Sigma^{-1} \bm{x}\right\Vert_{q}\right]\right)\\
         &+ \eta^2 \mathbb{E}_{\mathbb{P}_n} \left[ \bm{v}_{\Sigma ^{-1}\bm{x}}^\top  \Sigma^{-1}  \bm{v}_{\Sigma ^{-1}\bm{x}}\right]-\eta^2 \mathbb{E}_{\mathbb{P}_n} \left[ \bm{v}_{(\Sigma^{-1}+\frac{1}{\sqrt{n}} \bm{U})\bm{x}}^\top  \left(\Sigma^{-1}+\frac{1}{\sqrt{n}} \bm{U}\right)  \bm{v}_{(\Sigma^{-1}+\frac{1}{\sqrt{n}} \bm{U})\bm{x}}\right]\\
         &+o(1).
    \end{aligned}
\end{equation*}

It follows from the proof of Theorem  1 in \cite{yuan2007model} that 
\[-\log \det (\Sigma^{-1}+ \frac{1}{\sqrt{n}}\bm{U})+\log \det (\Sigma^{-1})= -\frac{\mathrm{tr}(\bm{U}\Sigma)}{\sqrt{n}}+\frac{1}{2}\frac{\mathrm{tr}(\bm{U}\Sigma \bm{U}\Sigma)}{n}+o\left(\frac{1}{n}\right).\]

We can rewrite the term $\mathbb{E}_{\mathbb{P}_n} \left[\bm{x}^\top  \bm{U} \bm{x}\right]$ as $\mathrm{tr}(\bm{U}\bar{A})$, where 
\[\bar{A}=\frac{1}{n}\sum_{i=1}^n \bm{x}_i\bm{x}_i^\top.\]

In this way, we have that 
\begin{equation}\begin{aligned}
&-n\log \det (\Sigma^{-1}+ \frac{1}{\sqrt{n}}\bm{U})+n\log \det (\Sigma^{-1})+ \sqrt{n}\mathbb{E}_{\mathbb{P}_n} \left[\bm{x}^\top  \bm{U} \bm{x}\right]\\
=& \mathrm{tr}(\bm{U}\sqrt{n}(\bar{A}-\Sigma)) +\frac{1}{2}\mathrm{tr}(\bm{U}\Sigma \bm{U}\Sigma)+o(1)\\
\Rightarrow & -\mathrm{tr}(\bm{U}\mathbf{G}) + \frac{1}{2}\mathrm{tr}(\bm{U}\Sigma \bm{U}\Sigma),
\end{aligned} \end{equation}
where $\mathrm{vec} (\mathbf{G})\sim \mathcal{N}(\bm{0}, \Lambda)$, and $\Lambda$ is such that 
\[\mathrm{cov}([\mathbf{G}]_{ij},[\mathbf{G}]_{i^\prime j^\prime})=\mathrm{cov}(\bm{x}_i\bm{x}_j, \bm{x}_{i^\prime}\bm{x}_{j^\prime}).\]

Notice that we also have
\begin{equation}\begin{aligned}\sqrt{n}\left(\mathbb{E}_{\mathbb{P}_n} \left[ \left\Vert (\Sigma^{-1}+\frac{1}{\sqrt{n}}\bm{U})\bm{x}\right\Vert_{q}\right]-\mathbb{E}_{\mathbb{P}_n}\left[ \left\Vert  \Sigma^{-1} \bm{x}\right\Vert_{q}\right]\right)\\\to_p \mathbb{E}\left[\frac{\left( \mathrm{sign}(\Sigma^{-1}\bm{x}) \odot |\Sigma^{-1}\bm{x}|^{q-1} \right)^\top \bm{U} \bm{x} }{\Vert\Sigma^{-1}\bm{x}\Vert_q^{q-1}}\right].  
\end{aligned}\end{equation}

Thus, 
\[V_n(\bm{U})\Rightarrow V(\bm{U}),\]
where 
\[ V(\bm{U}) = -\mathrm{tr}(\bm{U}\mathbf{G}) +\frac{1}{2}\mathrm{tr}(\bm{U}\Sigma \bm{U}\Sigma) + 2\eta\mathbb{E}\left[\frac{\left( \mathrm{sign}(\Sigma^{-1}\bm{x}) \odot |\Sigma^{-1}\bm{x}|^{q-1} \right)^\top \bm{U} \bm{x} }{\Vert\Sigma^{-1}\bm{x}\Vert_q^{q-1}}\right].\]

We take the gradient of $V(\bm{U})$ w.r.t. $\bm{U}$ and obtain  that 
\[ -\mathbf{G}+ \Sigma \bm{U}^\ast \Sigma  +2\eta \mathbb{E}\left[\frac{\left( \mathrm{sign}(\Sigma^{-1}\bm{x}) \odot |\Sigma^{-1}\bm{x}|^{q-1} \right) \bm{x}^\top }{\Vert\Sigma^{-1}\bm{x}\Vert_q^{q-1}}\right] = 0,\]
\[\bm{U}^\ast  = \Sigma^{-1}  
\left(\mathbf{G} -2\eta\mathbb{E}\left[\frac{\left( \mathrm{sign}(\Sigma^{-1}\bm{x}) \odot |\Sigma^{-1}\bm{x}|^{q-1} \right) \bm{x}^\top }{\Vert\Sigma^{-1}\bm{x}\Vert_q^{q-1}}\right]\right) \Sigma^{-1}. \]

Since $V(\bm{U})$ is strongly convex due to the positive definiteness of $\Sigma$, we have that 
\[ \sqrt{n}\left(\widehat{ C}-\Sigma^{-1} \right)\Rightarrow  \Sigma^{-1}  
\left(\mathbf{G} -2\eta\mathbb{E}\left[\frac{\left( \mathrm{sign}(\Sigma^{-1}\bm{x}) \odot |\Sigma^{-1}\bm{x}|^{q-1} \right) \bm{x}^\top }{\Vert\Sigma^{-1}\bm{x}\Vert_q^{q-1}}\right]\right) \Sigma^{-1}.\]

\noindent\textbf{Case 2}: $\gamma>1/2$.
\begin{equation*}
    \begin{aligned}
          &V_n(\bm{U}) \\\overset{\vartriangle}{=} & n\left(\Psi_n(\Sigma^{-1}+\frac{1}{\sqrt{n}}\bm{U})-\Psi (\Sigma^{-1})\right) \\         
          =& -n\log \det (\Sigma^{-1}+ \frac{1}{\sqrt{n}}\bm{U})+n\log \det (\Sigma^{-1})+ \sqrt{n}\mathbb{E}_{\mathbb{P}_n} \left[\bm{x}^\top  \bm{U} \bm{x}\right]\\
         &+2\eta n^{1/2-\gamma}\sqrt{n}\left(\mathbb{E}_{\mathbb{P}_n} \left[ \left\Vert (\Sigma^{-1}+\frac{1}{\sqrt{n}}\bm{U})\bm{x}\right\Vert_{q}\right]-\mathbb{E}_{\mathbb{P}_n}[ \left\Vert  \Sigma^{-1} \bm{x}\right\Vert_{q}]\right)\\
         &+ \eta^2 n^{1-2\gamma} \left(\mathbb{E}_{\mathbb{P}_n} \left[ \bm{v}_{\Sigma ^{-1}\bm{x}}^\top  \Sigma^{-1}  \bm{v}_{\Sigma ^{-1}\bm{x}}\right]- \mathbb{E}_{\mathbb{P}_n} \left[ \bm{v}_{(\Sigma^{-1}+\frac{1}{\sqrt{n}} \bm{U})\bm{x}}^\top  \left(\Sigma^{-1}+\frac{1}{\sqrt{n}} \bm{U}\right)  \bm{v}_{(\Sigma^{-1}+\frac{1}{\sqrt{n}} \bm{U})\bm{x}}\right]\right)\\
         &+o(1).
    \end{aligned}
\end{equation*}

Since $\gamma>1/2$, we have that 
\[V_n(\bm{U})\Rightarrow V(\bm{U}),\]
where 
\[ V(\bm{U}) = -\mathrm{tr}(\bm{U}\mathbf{G}) +\frac{1}{2}\mathrm{tr}(\bm{U}\Sigma \bm{U}\Sigma).\]

 Then, we have that 
 \[ \sqrt{n}\left(\widehat{ C}-\Sigma^{-1} \right)\Rightarrow  \Sigma^{-1}  
\mathbf{G}   \Sigma^{-1}.\]

\noindent\textbf{Case 3}: $0<\gamma<1/2$.
\begin{equation*}
    \begin{aligned}
          &V_n(\bm{U}) \overset{\vartriangle}{=} \\& n^{2\gamma}\left(\Psi_n(\Sigma^{-1}+\frac{1}{n^\gamma}\bm{U})-\Psi (\Sigma^{-1})\right) \\         
          =& -n^{2\gamma}\log \det (\Sigma^{-1}+ \frac{1}{n^\gamma}\bm{U})+n^{2\gamma}\log \det (\Sigma^{-1})+ n^{\gamma}\mathbb{E}_{\mathbb{P}_n} \left[\bm{x}^\top  \bm{U} \bm{x}\right]\\
         &+2\eta n^{\gamma}\left(\mathbb{E}_{\mathbb{P}_n} \left[ \left\Vert (\Sigma^{-1}+\frac{1}{n^{\gamma}}\bm{U})\bm{x}\right\Vert_{q}\right]-\mathbb{E}_{\mathbb{P}_n}[ \left\Vert  \Sigma^{-1} \bm{x}\right\Vert_{q}]\right)\\
         &+ \eta^2  \left(\mathbb{E}_{\mathbb{P}_n} \left[ \bm{v}_{\Sigma ^{-1}\bm{x}}^\top  \Sigma^{-1}  \bm{v}_{\Sigma ^{-1}\bm{x}}\right]-\mathbb{E}_{\mathbb{P}_n} \left[ \bm{v}_{(\Sigma^{-1}+\frac{1}{n^\gamma} \bm{U})\bm{x}}^\top  \left(\Sigma^{-1}+\frac{1}{n^\gamma} \bm{U}\right)  \bm{v}_{(\Sigma^{-1}+\frac{1}{n^\gamma} \bm{U})\bm{x}}\right]\right)\\
         &+o(1).
    \end{aligned}
\end{equation*}

We have that 
\[-n^{2\gamma}\log \det (\Sigma^{-1}+ \frac{1}{n^\gamma}\bm{U})+n^{2\gamma}\log \det (\Sigma^{-1})+ n^{\gamma}\mathbb{E}_{\mathbb{P}_n} \left[\bm{x}^\top  \bm{U} \bm{x}\right] \Rightarrow \frac{1}{2}\mathrm{tr}(\bm{U}\Sigma \bm{U}\Sigma).\]

We conclude that 
\[V_n(\bm{U})\Rightarrow V(\bm{U}),\]
where 
\[ V(\bm{U}) =\frac{1}{2}\mathrm{tr}(\bm{U}\Sigma \bm{U}\Sigma) + 2\eta\mathbb{E}\left[\frac{\left( \mathrm{sign}(\Sigma^{-1}\bm{x}) \odot |\Sigma^{-1}\bm{x}|^{q-1} \right)^\top \bm{U} \bm{x} }{\Vert\Sigma^{-1}\bm{x}\Vert_q^{q-1}}\right].\]

Then, we have that 
\[ \sqrt{n}\left(\widehat{ C}-\Sigma^{-1} \right)\Rightarrow -2 \eta\Sigma^{-1}  
\mathbb{E}\left[\frac{\left( \mathrm{sign}(\Sigma^{-1}\bm{x}) \odot |\Sigma^{-1}\bm{x}|^{q-1} \right) \bm{x}^\top }{\Vert\Sigma^{-1}\bm{x}\Vert_q^{q-1}}\right]\Sigma^{-1}.\]
\end{proof}

\section{Proof of Theorem \ref{asymptoticdistributionsurrogate}}
\begin{proof}
% We focus on the following optimization problem:
% \begin{equation*} \inf_{ C\succ 0}\left\{-\log \mathrm{det}  C + \mathbb{E}_{\mathbb{P}_n}\left[\bm{x}^\top  C\bm{x}\right]+ 2\delta\sum_{j=1}^d \sum_{k=1}^d \left(\frac{1}{n}\sum_{i=1}^n \vert\bm{x}^i_k\vert\right) \vert [ C]_{kj}\vert+ \delta^2 \Vert  C\Vert_{1,1}\right\} \end{equation*}

We define \[\Phi_n( C)= -\log \mathrm{det}  C + \mathbb{E}_{\mathbb{P}_n}\left[\bm{x}^\top  C\bm{x}\right]+ 2\delta\sum_{j=1}^d \sum_{k=1}^d \left(\frac{1}{n}\sum_{i=1}^n \vert\bm{x}^i_k\vert\right) \vert [ C]_{kj}\vert+\delta^2 \Vert  C\Vert_{1,1}.\]
For $ C\succ 0$, the function $\Phi_n( C)$  is convex.

\noindent\textbf{Case 1}: $\gamma=1/2$.
\begin{equation*}
    \begin{aligned}
          W_n(\bm{U}) \overset{\vartriangle}{=} &n\left(\Phi_n(\Sigma^{-1}+\frac{1}{\sqrt{n}}\bm{U})-\Phi (\Sigma^{-1})\right)\\     
          =&-n\log \det (\Sigma^{-1}+ \frac{1}{\sqrt{n}}\bm{U})+n\log \det (\Sigma^{-1})+ \sqrt{n}\mathbb{E}_{\mathbb{P}_n} \left[\bm{x}^\top  \bm{U} \bm{x}\right]\\
         &+2\eta\sqrt{n} \sum_{j=1}^d \sum_{k=1}^d \left(\frac{1}{n}\sum_{i=1}^n \vert\bm{x}^i_k\vert \right) \left(\vert [\Sigma^{-1}]_{kj}+\frac{1}{\sqrt{n}}\bm{U}_{kj}\vert -\vert [\Sigma^{-1}]_{kj}\vert\right)\\
         &+ o(1).
    \end{aligned}
\end{equation*}

Based on the proof of  Theorem \ref{asymptoticresults1}, we can have that 
\[\sqrt{n}(\widetilde{ C}-\Sigma^{-1})\Rightarrow\arg\min_{\bm{U}} W(\bm{U}),\]
where 
\begin{equation*}\begin{aligned} W(\bm{U}) &= -\mathrm{tr}(\bm{U}\mathbf{G}) +\frac{1}{2}\mathrm{tr}(\bm{U}\Sigma \bm{U}\Sigma)\\&+ 2\eta  \sum_{j=1}^d \sum_{k=1}^d  \mathbb{E}[\vert\bm{x}_k\vert] \left( \bm{U}_{jk}\mathrm{sign}([\Sigma^{-1}]_{jk})\mathbb{I}([\Sigma^{-1}]_{jk}\not = 0) +\vert  \bm{U}_{jk} \vert \mathbb{I}([\Sigma^{-1}]_{jk} = 0)\right).\end{aligned}\end{equation*}

\noindent\textbf{Case 2}: $\gamma>1/2$.
\begin{equation*}\begin{aligned} W(\bm{U}) &= -\mathrm{tr}(\bm{U}\mathbf{G}) +\frac{1}{2}\mathrm{tr}(\bm{U}\Sigma \bm{U}\Sigma).\end{aligned}\end{equation*}

\noindent\textbf{Case 3}: $0<\gamma<1/2$.
\begin{equation*}\begin{aligned} W(\bm{U}) &= \frac{1}{2}\mathrm{tr}(\bm{U}\Sigma \bm{U}\Sigma)\\&+ 2\eta  \sum_{j=1}^d \sum_{k=1}^d \mathbb{E}[\vert\bm{x}_k\vert] \left( \bm{U}_{jk}\mathrm{sign}([\Sigma^{-1}]_{jk})\mathbb{I}([\Sigma^{-1}]_{jk}\not = 0) +\vert  \bm{U}_{jk} \vert \mathbb{I}([\Sigma^{-1}]_{jk} = 0)\right).\end{aligned}\end{equation*}

Since we have $\bm{x}\sim \mathcal{N}(\bm{0},\Sigma)$, then we have $\bm{x}_k\sim \mathcal{N}(\bm{0}, \Sigma_{kk})$. In this way, 
\[\mathbb{E}[\vert\bm{x}_k\vert]= \sqrt{\frac{2\Sigma_{kk}}{\pi}}. \]
\end{proof}

\bibliographystyle{apalike}
\bibliography{JASA}
\end{document}